\documentclass[runningheads]{llncs}
\usepackage[T1]{fontenc}
\usepackage{graphicx}

\usepackage{amsfonts}
\usepackage{amsmath}
\usepackage{tikz}
\usepackage{hyperref}

\usepackage[noend]{algpseudocode}
\usepackage{algorithm,algorithmicx}

\algnewcommand\algorithmicupon{\textbf{upon}}
\algblockdefx[UPON]{Upon}{EndUpon}[1]
{\algorithmicupon\ #1\ \algorithmicdo}
{\algorithmicend\ \algorithmicupon}
\makeatletter
\ifthenelse{\equal{\ALG@noend}{t}}%
{\algtext*{EndUpon}}
{}%
\makeatother

\usepackage{etoolbox}
\usepackage{tikz}
\usetikzlibrary{tikzmark}
\usetikzlibrary{calc}
\usetikzlibrary{positioning}
\usepackage{soul}
\newcommand{\ALGtikzmarkcolor}{black}
\newcommand{\ALGtikzmarkextraindent}{4pt}
\newcommand{\ALGtikzmarkverticaloffsetstart}{-.5ex}
\newcommand{\ALGtikzmarkverticaloffsetend}{-.5ex}
\makeatletter
\newcounter{ALG@tikzmark@tempcnta}

\newcommand\ALG@tikzmark@start{%
	\global\let\ALG@tikzmark@last\ALG@tikzmark@starttext%
	\expandafter\edef\csname ALG@tikzmark@\theALG@nested\endcsname{\theALG@tikzmark@tempcnta}%
	\tikzmark{ALG@tikzmark@start@\csname ALG@tikzmark@\theALG@nested\endcsname}%
	\addtocounter{ALG@tikzmark@tempcnta}{1}%
}

\def\ALG@tikzmark@starttext{start}
\newcommand\ALG@tikzmark@end{%
	\ifx\ALG@tikzmark@last\ALG@tikzmark@starttext
	\else
	\tikzmark{ALG@tikzmark@end@\csname ALG@tikzmark@\theALG@nested\endcsname}%
	\tikz[overlay,remember picture] \draw[\ALGtikzmarkcolor] let \p{S}=($(pic cs:ALG@tikzmark@start@\csname ALG@tikzmark@\theALG@nested\endcsname)+(\ALGtikzmarkextraindent,\ALGtikzmarkverticaloffsetstart)$), \p{E}=($(pic cs:ALG@tikzmark@end@\csname ALG@tikzmark@\theALG@nested\endcsname)+(\ALGtikzmarkextraindent,\ALGtikzmarkverticaloffsetend)$) in (\x{S},\y{S})--(\x{S},\y{E});%
	\fi
	\gdef\ALG@tikzmark@last{end}%
}

\apptocmd{\ALG@beginblock}{\ALG@tikzmark@start}{}{\errmessage{failed to patch}}
\pretocmd{\ALG@endblock}{\ALG@tikzmark@end}{}{\errmessage{failed to patch}}
\makeatother

\usepackage[misc]{ifsym}

\begin{document}
\sloppy
\title{On the Solvability of Byzantine-tolerant Reliable Communication in Dynamic Networks\thanks{Corresponding author: Giovanni Farina - \url{giovanni.farina@unicusano.it}
\\\emph{This work is dedicated to Alexandre Maurer, whose work sparked this research journey.}
}}
\titlerunning{Reliable Communication in Dynamic Networks}
%
\author{Silvia Bonomi\inst{1}\orcidID{0000-0001-9928-5357} \and
Giovanni Farina\inst{2}\orcidID{0000-0002-4792-5305}
\and
S\'{e}bastien Tixeuil\inst{3,4}\orcidID{0000-0002-0948-7172}}
\authorrunning{S. Bonomi et al.}
%
\institute{Sapienza University of Rome, Italy
\and
Department of Engineering, Niccolò Cusano University, Italy
\and
Sorbonne Université, CNRS, LIP6
\and
Institut Universitaire de France
}
\maketitle              
\begin{abstract}
A reliable communication primitive guarantees the delivery, integrity, and authorship of messages exchanged between correct processes of a distributed system.  
We investigate the necessary and sufficient conditions for reliable communication in dynamic networks, where the network topology evolves over time despite the presence of a limited number of Byzantine faulty processes that may behave arbitrarily (i.e., in the globally bounded Byzantine failure model).  
We identify classes of dynamic networks where such conditions are satisfied, and extend our analysis to message losses, local computation with unbounded finite delay, and authenticated messages. 

\end{abstract}



%
%
%
\setcounter{page}{1}

\begin{center}
    \small \textit{This technical report presents the complete version of the paper published in the SIROCCO 2026 conference proceedings \cite{10.1007/978-3-032-26465-7_7}.}
\end{center}

\section{Introduction}
The reliable communication primitive is a fundamental building block for distributed systems. 
It guarantees the proper exchange of messages between correct processes, even if they are not directly connected by a communication link or if some participating processes behave maliciously (i.e., are Byzantine).
In particular, the primitive ensures the \textit{authorship}, \textit{integrity}, and \textit{delivery} of the information exchanged between correct processes. 
More precisely, it mandates that \emph{(i)} all messages exchanged between correct processes are not modified during their propagation (integrity), \emph{(ii)} messages eventually reach their destination (delivery), and \emph{(iii)} the author of a message cannot be forged (authorship).
Implementations of the reliable communication primitive have been studied in several settings~\cite{DBLP:conf/sss/BonomiFT18,DBLP:journals/jbcs/BonomiFT19,DBLP:conf/focs/Dolev81,DBLP:conf/wdag/MaurerT12,DBLP:conf/srds/MaurerTD15,DBLP:journals/networks/Pelc96,DBLP:journals/ipl/PelcP05}. 
In this paper, we consider the \textit{globally bounded Byzantine failure model} (i.e., there is an upper bound $f$ on the number of Byzantine faulty processes, and $f$ is known to all correct processes), and a \textit{dynamic communication network}. 
To the best of our knowledge, the work of Maurer et al.~\cite{DBLP:conf/srds/MaurerTD15} is the only one that analyzes the solvability of the primitive in such a setting. 
It characterizes the necessary and sufficient conditions for reliable communication from a specific source to a defined target at a given time. 
However, the verification of the conditions identified by Maurer et al.~\cite{DBLP:conf/srds/MaurerTD15} is equivalent to solving a NP-complete problem~\cite{DBLP:journals/tcs/FluschnikMNRZ20,DBLP:journals/jcss/KempeKK02}.

In this work, we first extend such conditions to characterize when \emph{all} correct processes can achieve reliable communication in a synchronous system \emph{independently of the diffusion time}. 
We identify classes of dynamic networks that \emph{(i)} satisfy the solvability conditions \emph{(ii)} and where verifying that a network belongs to these classes is polynomial in the number of processes in the system.
We then expand our analysis to include an asynchronous system with lossy links and find that the solvability conditions for synchronous systems also apply to asynchronous ones.
Finally, we examine the case of authenticated messages.

\section{Related Work}
The reliable communication problem, which aims to implement a primitive to guarantee the delivery, integrity, and authorship of messages exchanged in a distributed system, has been the subject of extensive research in static networks under various assumptions. These include various types of faults, such as omissions, crashes, and arbitrary faults; alternative fault distributions, including deterministic and probabilistic models; and the impact of these faults on links ~\cite{DBLP:journals/networks/Pelc92}, processes, or both.
In a seminal work, Dolev~\cite{DBLP:conf/focs/Dolev81} identified the necessary and sufficient condition for this problem in static networks with reliable and authenticated links, assuming an upper bound $f$ on the number of Byzantine faulty processes. The condition states that the network must be $(2f+1)$-connected to tolerate $f$ such processes.

Subsequent work investigated more constrained process failure distributions. 
Koo~\cite{koo2004broadcast} analyzed the reliable communication problem assuming only a fraction of nodes in a process's neighborhood can be compromised. 
Pelc and Peleg~\cite{DBLP:journals/ipl/PelcP05} generalized Koo's results by characterizing a failure model with an upper bound on faulty processes in each node's neighborhood and defining a corresponding solution.  
Pagourtzis et al.~\cite{DBLP:journals/dc/PagourtzisPS17} further generalized these models by assuming non-homogeneous local bounds on the number of faulty processes in each node's neighborhood. 
They also showed this failure model can be extended by assuming specific sets of potentially faulty processes (the \textit{general adversary model}), and that knowledge of the network topology increases the number of faulty processes that can be tolerated in the non-homogeneous, locally bounded setting.

Most solutions to the reliable communication problem rely on node-disjoint path redundancy, thus requiring highly connected networks. 
Consequently, several works have investigated weaker Byzantine-tolerant reliable communication primitives for use in loosely connected networks. These primitives may allow a small minority of correct processes to either deliver invalid messages or fail to deliver genuine ones~\cite{DBLP:journals/jpdc/MaurerT14,DBLP:journals/tpds/MaurerT15,DBLP:journals/ppl/MaurerT16}.

Over the past two decades, the integration of computing and communication capabilities into a growing number of devices (\emph{e.g.}, IoT networks) has spurred increased interest in the modeling and analysis of dynamic distributed systems.
Most of the dynamic distributed system models proposed so far can be categorized as either \textit{open} or \textit{closed}. The former involves new processes continuously entering and leaving the system (a phenomenon often referred to as \textit{churn}), while the latter assumes a fixed set of processes whose communication links may change over time.
In the context of closed dynamic networks, Maurer et al.~\cite{DBLP:conf/srds/MaurerTD15} identified the necessary and sufficient conditions for solving a single instance of the reliable communication problem in the presence of a subset of $f$ Byzantine faulty processes. 
Bonomi et al.~\cite{DBLP:conf/sss/BonomiFT18,BONOMI2024104952} identified the solvability conditions assuming the homogeneous locally bounded failure model of Pelc and Peleg~\cite{DBLP:journals/ipl/PelcP05}.
Maurer~\cite{DBLP:conf/opodis/Maurer20} later extended the primitive to also withstand transient process failures.

Cryptography (specifically, digital signatures), by ensuring the authenticity and integrity of exchanged information, can be used to achieve Byzantine-tolerant reliable communication~\cite{castro1999practical,drabkin2005efficient}. 
The main advantage of such cryptographic tools is that they enable reliable communication with simpler solutions and under weaker connectivity requirements. 
However, their correctness depends on the underlying cryptosystem and the assumption that the adversary has bounded computing power.
A common assumption of Byzantine-tolerant reliable communication protocols is the use of authenticated point-to-point channels, which prevents a process from impersonating multiple others (a Sybil attack)~\cite{DBLP:conf/iptps/Douceur02}.
The primary distinction between cryptographic (authenticated) and non-cryptographic (unauthenticated) protocols for reliable communication lies in how cryptographic primitives are used: non-cryptographic protocols can only use digital signatures between neighbors for authentication, while cryptographic protocols use them to allow message verification even between nodes that are not directly connected. 
Finally, it is worth noting that cryptography is not strictly required to implement an authenticated channel~\cite{DBLP:journals/wc/ZengGM10}.

This work builds upon the seminal characterization of the reliable communication problem in closed dynamic networks by Maurer et al.~\cite{DBLP:conf/srds/MaurerTD15}, extending it in several directions by: \textit{(i)} specifying the solvability conditions for any pair of processes at any time; \textit{(ii)} identifying classes of dynamic networks where these conditions hold; and \textit{(iii)} considering weaker or alternative models in which messages can be lost and local computation delays are unknown.

\section{Graphs and Evolving Graphs}
A static undirected \textit{graph} is a pair $G = (V,E)$ of sets $V$ and $E$ 
such that $E \subseteq \binom{{V}}{2}$. 
The elements $p_i$ of $V$ are the \textit{vertices} (or \textit{nodes}) of the graph, whereas the elements of $E$ are the (undirected) \textit{edges}~\cite{Diestel_2017}. 
Two vertices $p_s, p_t$ connected by an edge $\{p_s,p_t\} \in E$ are called \textit{neighbors}.
A sequence of distinct nodes $P = (p_1, p_2, \dots, p_m)$ constitutes a \textit{path} if every pair of consecutive nodes $p_i$ and $p_{i+1}$ in the sequence satisfies the condition that they are neighbors. Nodes $p_1$ and $p_m$ in $P$ are referred to as \textit{endpoints}.
Two or more paths are \emph{disjoint} if they share no node except their endpoints.
A graph is \textit{connected} if there exists a path between every pair of nodes.
A graph is \textit{$k$-connected} if removing any subset $S \subset V$ of $k-1$ nodes (and all the edges having at least one node in $S$) results in a connected subgraph. 
If a graph is $k$-connected, then there exists a set of $k$ disjoint paths between all pairs of nodes~\cite{Menger1927}.
The \textit{node connectivity} of a graph is the minimum number of nodes that have to be removed from the graph to disconnect it~\cite{Diestel_2017}.


An \textit{evolving graph} (\emph{a.k.a.} temporal graph) $\mathcal{G} = (G_0, G_1, \dots, G_j, \dots)$~\cite{DBLP:journals/network/Ferreira04} is  a sequence of static undirected graphs, where each graph $G_j := (V,E_j \subseteq \binom{{V}}{2})$ denotes a \emph{snapshot} of $\mathcal{G}$.
All snapshots share the same set of vertices $V$, whereas the set of edges may change at every snapshot. 
An evolving graph is defined over a set of time instants $\mathcal{T} \subseteq \mathbb{N}$, potentially infinite, drawn from the natural numbers, called the \textit{lifetime} of $\mathcal{G}$.
Specifically, every snapshot $G_j \in \mathcal{G}$ is associated with time instant $j \in \mathcal{T}$, and vice-versa.
We say that a particular edge $e$ is \textit{present} (or \textit{appears}) at time $j$ if $e \in E_{j}$ (\emph{i.e.}, $e$ is among the edges of snapshot $G_j$). 
The graph $\mathbb{G} = (V, E = \bigcup_{{j}\in \mathcal{T}} E_j)$ is the \textit{underlying graph}. We refer to $V$ and $E$ as the \textit{vertex set} and \textit{edge set}, respectively, of the evolving graph.

Given a subset $T \subseteq \mathcal{T}$, a \textit{temporal subgraph} $\mathcal{G}_{T}$ of $\mathcal{G}$ is the evolving graph that restricts the lifetime of $\mathcal{G}$ to the instants $j \in T$, namely, $\mathcal{G}_{T}$ contains only the snapshots $G_j$ in $\mathcal{G}$ such that $j \in T$. 
Then, $\mathcal{G}_{[x,y]}$ is the temporal subgraph induced by the contiguous interval $[x,y]$ of $\mathcal{T}$, i.e., $\mathcal{G}_{[x,y]}$ limits the lifetime of $\mathcal{G}$ to the period $[x, y]$. 
Similarly, the temporal subgraph $\mathcal{G}_{[x,*]}$ restricts the lifetime of $\mathcal{G}$ to the time instants after $x$ ($x$ included).
A \textit{spatial subgraph} $\mathcal{G}[\bar{V}, \bar{E}]$ is the evolving graph resulting from $\mathcal{G}$ when considering the set $\bar{V} \subseteq V$ as vertex set, and $\bar{E} \subseteq (E \cap \binom{\bar{V}}{2})$ as edge set. 
We specify a spatial subgraph only by its vertex set $\bar{V}$, namely $\mathcal{G}[\bar{V}]$, if its edge set $\bar{E}$ consists of all the edges in the edge set that have both endpoints in $\bar{V}$, i.e. $\bar{E} = E \cap \binom{\bar{V}}{2}$.
A \textit{spatial temporal} subgraph $\mathcal{G}[\bar{V}]_{T}$ of $\mathcal{G}$ considers a subset $T$ of the lifetime $\mathcal{T}$, a subset $\bar{V}$ of the vertices ${V}$ as vertex set, and $\bar{E} = E \cap \binom{\bar{V}}{2}$
as edge set.
Figure~\ref{fig:tvg_example} shows a graphical example of an evolving graph spanning over three time instants, while Figure~\ref{fig:sttvg_example} presents a spatial temporal subgraph of the former, removing node $p_2$ and snapshot $G_2$.

A \textit{journey} is the analogue of a path in an evolving graph.


\begin{definition}[Journey~\cite{DBLP:conf/srds/MaurerTD15}\footnote{Any definition or theorem followed by a reference has been defined or proven in that reference. All other definitions and theorems are novel to this work.\\This definition has been simplified due to the absence of the latency function in the model we adopted, and has been adapted to a discrete lifetime.}]
    Given an evolving graph $\mathcal{G} = (G_0, G_1, \dots, G_j, \dots)$, where each $G_j = (V, E_j)$, and two of its nodes $p_1, p_m \in V$,
    a \emph{journey} from $p_1$ to $p_m$, denoted with $p_1 \rightsquigarrow p_m$, is a pair  $J = (A,B)$ of two sequences $A$ and $B$ such that: $A$ is
    a sequence of distinct nodes $(p_1, p_2, \dots, p_m)$ 
    and $B$ is a strictly increasing sequence of time instants $(j_I, j_{II}, \dots, j_{m-1})$, with $j_\iota \in \mathcal{T}$, such that for all $i \in \{1, \dots, m-1\}$, it holds that $\{p_i, p_{i+1}\} \in E_{j_\iota }$.

\end{definition}


\begin{definition}[Set $\Sigma(\mathcal{G},p_s,p_t)$ of node sets between two vertices~\cite{DBLP:conf/srds/MaurerTD15}]
    Given an evolving graph $\mathcal{G}$ and two of its nodes $p_s, p_t \in V$, $\Sigma(p_s,p_t)$ refers to the set of node sets $\{\{p_1, \dots, p_l\} | (p_s, p_1, \dots, p_l, p_t)$ is a journey in $\mathcal{G}\}$.
\end{definition}

\begin{definition}[Hitting set~\cite{DBLP:books/fm/GareyJ79}]
    Let $\Omega = \{S_1, S_2, \dots, S_m\}$ be a family of subsets of a universe $U$, where $S_i \subseteq U$ for each $i \in \{1, \dots, m\}$. 
    A set $H \subseteq U$ is called a \emph{hitting set} for $\Omega$ if it intersects every set in $\Omega$, that is,
\[
\forall i \in \{1, \dots, m\}, \quad H \cap S_i \neq \emptyset.
\]

\end{definition}

\begin{definition}
    [Minimum Hitting Set~\cite{DBLP:books/fm/GareyJ79} and {MinCut}~\cite{DBLP:conf/srds/MaurerTD15}]
    Let $\Omega = \{S_1, S_2, \dots, S_m\}$ be a family of subsets of a universe $U$, where $S_i \subseteq U$ for each $i \in \{1, \dots, m\}$. A set $H \subseteq U$ is called a \emph{minimum hitting set} for $\Omega$ if:
\begin{itemize}
  \item $H$ is a hitting set for $\Omega$, i.e., $H \cap S_i \neq \emptyset$ for all $i \in \{1, \dots, m\}$, and
  \item $H$ has the smallest possible cardinality among all hitting sets for $\Omega$, i.e., for every hitting set $H' \subseteq U$, it holds that $|H| \leq |H'|$.
\end{itemize}

  We refer with $MinCut(\Omega)$ to the cardinality of a minimum hitting set for $\Omega$.

\end{definition}




\begin{definition}[Dynamic minimum cut size $k$ between two nodes and $p_s \rightsquigarrow_k p_t$ ~\cite{DBLP:conf/srds/MaurerTD15}]
    Given an evolving graph $\mathcal{G}$, two of its nodes $p_s, p_t \in V$, the dynamic minimum cut size $k$ from $p_s$ to $p_t$ is the smallest number of other nodes to remove from $\mathcal{G}$ so that no journey exists from $p_s$ to $p_t$.
    Specifically, $$DynMinCut(\mathcal{G},p_s,p_t) = MinCut(\Sigma(\mathcal{G},p_s,p_t))$$
    Note that if $\exists j \in \mathcal{T} : \{p_s, p_t\} \in E_j$, then $DynMinCut(\mathcal{G},p_s,p_t) = \infty$. 

    We denote with $p_s \rightsquigarrow_k p_t$ a set of journeys in $\mathcal{G}$ from $p_s$ to $p_t$ having a dynamic minimum cut size at least $k$.
\end{definition}


In the same way, various families of graph have been defined in graph theory (trees, planar graphs, grids, etc.), and several classes of evolving graphs have been characterized in the literature~\cite{DBLP:journals/paapp/CasteigtsFQS12,DBLP:books/hal/Casteigts18}. Specifically, a class of evolving graphs groups all graphs that satisfy a specific set of properties. We recall or define some relevant classes to our work in Section \ref{sec:classes}.

\begin{figure}
    \begin{minipage}[t]{0.5\linewidth}
        \centering
        \begin{minipage}[b]{0.32\linewidth}
            \centering
            \begin{tikzpicture}[main/.style = {draw, circle}] 
                \node[main] (1) {$1$}; 
                \node[main] (2) [above right of=1] {$2$};
                \node[main] (3) [below right of=1] {$3$}; 
                \node[main] (4) [above right of=3] {$4$};
                \draw (1) -- (2);
                \draw (1) -- (3);
            \end{tikzpicture} \\
            \vspace{2pt}
            (a) $1 \in \mathcal{T}$
        \end{minipage}
        \begin{minipage}[b]{0.32\linewidth}
        \centering
            \begin{tikzpicture}[main/.style = {draw, circle}] 
                \node[main] (1) {$1$}; 
                \node[main] (2) [above right of=1] {$2$};
                \node[main] (3) [below right of=1] {$3$}; 
                \node[main] (4) [above right of=3] {$4$};
                \draw (1) -- (2);
                \draw (2) -- (4);
            \end{tikzpicture}\\
            \vspace{2pt}
            (b) $2\in \mathcal{T}$
        \end{minipage}
        \begin{minipage}[b]{0.32\linewidth}
        \centering
                \begin{tikzpicture}[main/.style = {draw, circle}] 
                    \node[main] (1) {$1$}; 
                    \node[main] (2) [above right of=1] {$2$};
                    \node[main] (3) [below right of=1] {$3$}; 
                    \node[main] (4) [above right of=3] {$4$};
                    \draw (1) -- (3);
                    \draw (2) -- (4);
                    \draw (3) -- (4);
                \end{tikzpicture}\\
                \vspace{2pt}
                (c) $3\in \mathcal{T}$
        \end{minipage}
        \caption{An evolving graph example.}
        \label{fig:tvg_example}
    \end{minipage}%
    \begin{minipage}[t]{0.5\linewidth}
        \centering
        \begin{minipage}[b]{0.32\linewidth}
            \centering
            \begin{tikzpicture}[main/.style = {draw, circle}] 
                \node[main] (1) {$1$}; 
                \node[main] (3) [below right of=1] {$3$}; 
                \node[main] (4) [above right of=3] {$4$};
                \draw (1) -- (3);
            \end{tikzpicture} \\
            \vspace{2pt}
            (a) $1 \in \mathcal{T}$
        \end{minipage}
        \begin{minipage}[b]{0.32\linewidth}
        \centering
                \begin{tikzpicture}[main/.style = {draw, circle}] 
                    \node[main] (1) {$1$}; 
                    \node[main] (3) [below right of=1] {$3$}; 
                    \node[main] (4) [above right of=3] {$4$};
                    \draw (1) -- (3);
                    \draw (3) -- (4);
                \end{tikzpicture}\\
                \vspace{2pt}
                (b) $3 \in \mathcal{T}$
        \end{minipage}
        \caption{A spatial temporal subgraph of the evolving graph in Figure \ref{fig:tvg_example}}
        \label{fig:sttvg_example}
    \end{minipage}
\end{figure}

\section{System Model and Problem Statement}
We consider a distributed system composed of a fixed set of $n$ processes $N = \{p_1, p_2 \dots, p_n\}$, each one associated with a unique integer identifier.
The evolution of the system is characterized by events occurring at specific times defined by a fictitious global clock spanning the natural numbers $\mathbb{N}$.

Processes can communicate with each other by exchanging \emph{messages} over a \emph{dynamic communication network} composed of \emph{point-to-point links}. 
The dynamic communication network is modeled by an {evolving graph} $\mathcal{G} = (G_0, G_1, \dots, G_j, \dots)$. Each \emph{snapshot} $G_j = (V, ~ E_j)$ corresponds to the actual communication network at time $j$ where $V = N$ represents the set of processes participating in the system and $E_j \subseteq N \times N$ is the actual set of \emph{existing} (\textit{present}) edges at time $j$ (i.e., communication links available for the point-to-point communication).
In the following, we will interchangeably use the terms \textit{node} and \textit{process}, and the terms \textit{link} and \textit{edge}.
The evolving graph characterizes the communication network for the entire lifetime of the system.
At each time $j$, processes can only communicate by using present links, i.e., they can send/receive messages to/from their neighbors in the snapshot $G_j$. 
We refer to \emph{multicast} when a process sends a message through all of its available links, namely, to all of its current neighbors.
Each message $m$ is associated with a \textit{source/author} and a \textit{sender}: the source is the process \textit{that generates} message $m$, the sender is the process \textit{that relays} message $m$ through a link. 
The source and sender of a message may coincide.

At each time instant $j$, every process executes the following concurrent steps: \emph{compute}, \emph{send}, and \emph{receive}. Specifically, it carries out local computation and prepares (enqueues) messages to be sent to its neighbors, transmits these messages (if any), and receives messages from other processes (if any). 
Consequently, a message is able to traverse only one link during a single time instant $j$.

Processes execute a distributed protocol $\mathcal{P}$.
Processes can be either \textit{correct} or \textit{Byzantine faulty}. 
A correct process executes the protocol $\mathcal{P}$, Byzantine faulty processes can behave arbitrarily instead. 
In particular, while running $\mathcal{P}$, faulty processes can send arbitrary messages or omit to send/receive all or part of them. We assume that at most $f$ processes can be faulty (globally bounded Byzantine failure model).

To model the asynchrony and delays that may occur in the system, we consider two types of point-to-point communication primitives that characterize the behavior of the links (PL and FLL), and two alternative settings for the local computation delay of processes (SC and AC):
\begin{itemize}
    \item \textit{perfect link} (PL): it provides the \textit{reliable delivery} property~\cite{DBLP:books/daglib/0025983}, namely that if a correct process $p_s$ sends a message $m$ to a correct process $p_r$ at time $j$ via the present link, then $p_r$ receives $m$ at the same time $j$;
    %
    \item \textit{fair-loss link} (FLL): it guarantees the \textit{fair-loss} property~\cite{DBLP:books/daglib/0025983}, namely that if a correct process $p_s$ infinitely often (at distinct times) sends a message $m$ to a correct process $p_r$ via the present link, then $p_r$ receives $m$ an infinite number of times;
    \item \textit{synchronous computation} (SC): the local computation time (i.e., the time to perform the compute step) is negligible and is assumed to be equal to $0$;
    \item \textit{asynchronous computation} (AC): the local computation (i.e., the compute step) takes a finite and unknown amount of time.
\end{itemize}

\noindent For the sake of modeling, we assume that \textit{(i)} link properties only hold when links are present, that \textit{(ii)} links have \emph{unbounded capacity} (namely, if a correct process $p_s$ sends an arbitrary set $M$ of messages to a correct process $p_r$ at time $j$, then $p_r$ receives all messages in $M$ at the same time instant $j$ (PL or FLL) or they are lost (FLL)), and \textit{(iii)} that processes can instantly detect whether their links are available or not.
The set of assumptions PL and SC characterizes a synchronous distributed system where processes and links are able to accommodate every message propagation enabled by the dynamic communication network $\mathcal{G}$, namely, every journey in $\mathcal{G}$ represents a feasible propagation pattern for a message between its endpoints (for an example, the journey $(\left(p_s,p_r,p_t\right),(1,2))$ models the possible propagation of a message exchanged at time $1$ from $p_s$ to $p_r$ and then forwarded by $p_r$ to $p_t$ at time $2$). 
It follows that a message generated by a process $p_s$ at time $j$ in a PL-SC system can potentially reach any process $p_t$ such that there exists a journey from $p_s$ to $p_t$ in $\mathcal{G}_{[j,*]}$.
The FLL assumption characterizes potentially asynchronous and lossy links in this setting, meaning that a process may attempt a finite but unknown number of times to send a message on a link that is infinitely often present before its first reception.
The AC assumption characterizes the potential inability of processes to perform timely computation and generate messages to send for a finite amount of time.
Note that the AC assumption does not prevent processes from receiving messages.

Finally, we consider two alternative settings that limit the capability of faulty processes:
\begin{itemize}
    \item \textit{authenticated links} (AL): the identity of the sender of a message $m$, i.e., the process sending $m$ through a point-to-point link, cannot be forged;
    \item \textit{authenticated messages} (AM): the identity of the author of a message $m$, namely, the process that generated a message $m$, cannot be forged.
\end{itemize}

\noindent Note that the AM assumption can guarantee the authenticity of the author of a message over multiple links, while the AL assumption can guarantee the authenticity of the author only if it matches the sender of the message, and that the authenticated message setting does not implicitly assume authenticated links.

In the rest of the paper, we will indicate the setting under consideration, in terms of links and local computation assumptions, by specifying a triple $\langle \alpha, \beta, \gamma \rangle$ where $\alpha \in \{PL, FLL\}$, $\beta \in \{SC, AC\}$ and $\gamma \in \{AL, AM\}$.

\subsection*{The Reliable Communication Problem}
We aim at analyzing the \textit{reliable communication problem}, whose goal is the definition of a communication primitive that allows correct processes, not directly connected by a link, to exchange \textit{contents} guaranteeing their authorship, integrity, and delivery.

Let us denote as \textit{source} or \textit{author} the process $p_s$ that generates a content and as \textit{target} $p_t$ the peer to which such content is addressed.
A Reliable Communication ({\texttt{RC}) primitive is accessible by every process in the system and exposes two operations: \mbox{\texttt{{RC.send($p_t$, $c$)}}} and \texttt{RC.deliver($p_s$, $c$)}. 
The \texttt{send} operation is invoked by the source to disseminate a content, and the \texttt{deliver} operation notifies a process about the delivery of a content.\\
A protocol $\mathcal{P}$ implements a reliable communication primitive if it satisfies the following properties:
\begin{itemize}
	\item \textbf{\textit{safety}}: if $p_t$ is a \textit{correct} process and delivers a content $c$ from $p_s$, then $p_s$ previously sent $c$;
	\item \textbf{\textit{liveness}}: if $p_s$ is a \textit{correct} process and sends a content $c$ to a correct process $p_t$, then $p_t$ eventually delivers $c$ from $p_s$.
\end{itemize}

For the sake of notation, we denote by \textit{content} the payload exchanged by a reliable communication primitive, and by \textit{message} the unit of information exchanged by a distributed protocol over a point-to-point link.

We say that an \textit{instance} of the reliable communication problem \textit{starts at time $j$} (or simply \textit{at time} $j$) if at time $j$ the \texttt{RC.send} operation is executed, and that it \textit{terminates at time $j$} if the target process $p_t$ executes the \texttt{RC.deliver($p_s$, $c$)} operation for the first time at time $j$.

\section{Solvability Conditions}
We characterize the solvability of the reliable communication problem in several settings. 
We start by recalling the result of Maurer et al.~\cite{DBLP:conf/srds/MaurerTD15} who established the necessary and sufficient conditions for a \textit{one-to-one} (i.e., from a defined source to a fixed target) reliable communication starting at time $j$, in a setting where the links are perfect and authenticated and the computation is synchronous.
We then extend their conditions to any starting time and any pair of processes.
Afterwards, we look for further classes of evolving graphs where the learned conditions are verified.
Interestingly, the stricter enabling class we identify for the PL and SC setting remains minimal also for the weakest setting we consider, namely asynchronous computation and fair-loss links; 
another class instead, allows to upper bound the latency of any reliable communication instance when synchronous computation and perfect links are assumed. 
Finally, we extend all of our results to the settings where messages are authenticated, and we provide an analysis on the computational complexity of asserting class membership for the dynamic network classes we identified.
{All the results that follow, cited or provided, are based on the existence of a specific set of journeys in the communication network that support the propagation of a content in the considered settings.
Note that past theorems and definitions may be rephrased to fit our notations.}

Some of the intermediate results and proofs are reported in Appendix \ref{sec:results}.

\subsection{Classes of Evolving Graphs}
\label{sec:classes}

In this paper, we consider various classes of evolving graphs. While Definitions~\ref{def:temporalreach}-\ref{def:recedges} were previously considered in the literature~\cite{DBLP:journals/paapp/CasteigtsFQS12}, Definitions~\ref{def:noname0}-\ref{def:reckedges} are newly defined to address the presence of Byzantine processes.

\begin{definition}[Class $\mathcal{J}_{(s,t)}$ - Temporal Reachability]
    \label{def:temporalreach}
    The class $\mathcal{J}_{(s,t)}$ is the set of all evolving graphs $\mathcal{G}$ where there exists a journey from node $p_s$ to node $p_t$.
\end{definition}

\begin{definition}[Class $\mathcal{TC}$ - Temporal Connectivity]
    \label{def:temporalconnectivity}
    The class $\mathcal{TC}$ is the set of all evolving graphs $\mathcal{G}$ where $\forall p_s,p_t \in V$, there exists $(p_s \rightsquigarrow p_t) \in \mathcal{G}$ (the class of evolving graphs where a journey exists between every pair of nodes).
\end{definition}

\begin{definition}[Class $\mathcal{J}^\mathcal{R}_{(s,t)}$ - Recurrent Reachability]
    \label{def:rectemporalreach}
    The class $\mathcal{J}^\mathcal{R}_{(s,t)}$ is the set of all evolving graphs $\mathcal{G}$ where $\forall j \in \mathcal{T}$, $\mathcal{G}_{[j,*]} \in \mathcal{J}_{(s,t)}$.
\end{definition}

\begin{definition}[Class $\mathcal{TC}^\mathcal{R}$ - Recurrent temporal connectivity]
	\label{def:tcr} 
    The class $\mathcal{TC}^\mathcal{R}$ is the set of all evolving graphs $\mathcal{G}$ with infinite lifetime where $\forall j \in \mathcal{T}$, $\mathcal{G}_{[j,*]} \in \mathcal{TC}$ (the class where, for every time instant $j \in \mathcal{T}$, the temporal subgraph $\mathcal{G}_{[j,*]}$ is temporally connected).
\end{definition}

\begin{definition}[Class $\mathcal{C^*}$ - {Always-connected snapshots}, or {1-interval connectivity}]
	\label{def:oneinterval} 
    The class $\mathcal{C^*}$ is the set of all evolving graphs $\mathcal{G}$ where $\forall G_j \in \mathcal{G}$, $G_j$ is connected (the class where every snapshot is a connected graph).
\end{definition}

\begin{definition}[Class $\mathcal{E}^{\mathcal{R}}$ - Recurrent Edges]
	\label{def:recedges}
    The class $\mathcal{E}^{\mathcal{R}}$ is the set of all evolving graphs $\mathcal{G}$ with infinite lifetime where $\forall e \in  E, ~ \forall j \in \mathcal{T}, ~ \exists l > j , e \in E_l$ (the class where every edge is present infinitely often).
\end{definition}

\begin{definition}[Class $\mathcal{J}_{(s,t,k)}$ - $k$-journeys from $p_s$ to $p_t$]
	\label{def:noname0}
    The class $\mathcal{J}_{(a,b,k)}$ is the set of all evolving graphs $\mathcal{G}$ where $\exists (p_s \rightsquigarrow_k p_t) \in \mathcal{G}$ (the class where there exists a set of journeys with a dynamic minimum cut of size at least $k$ from node $p_s$ to node $p_t$).
\end{definition}

\begin{definition}[Class $\mathcal{J}_{(s,t,k)}^\mathcal{R}$ - Recurrent $k$-journeys from $p_s$ to $p_t$]
    \label{def:noname1} 
    The class $\mathcal{J}_{(s,t,k)}^\mathcal{R}$ is the set of all evolving graphs $\mathcal{G}$ with infinite lifetime where $\forall j \in \mathcal{T}$, $\mathcal{G}_{[j,*]} \in \mathcal{J}_{(s,t,k)}$ (the class where the temporal subgraph $\mathcal{G}_{[j,*]}$ is in $\mathcal{J}_{(s,t,k)}$ for every time $j\in \mathcal{T}$).
\end{definition}

\begin{definition}[Class $\mathcal{TC}_k$ - Temporal $k$-Connectivity]
    \label{def:tempKconn} 
    The class $\mathcal{TC}_k$ is the set of all evolving graphs $\mathcal{G}$ where $\forall p_s,p_t \in  V, ~ \exists (p_s \rightsquigarrow_k p_t) \in \mathcal{G}$ (the class where there exists a set of journeys having a dynamic minimum cut of size at least $k$ between every pair of nodes).
\end{definition}

\begin{definition}[Class $\mathcal{TC}^\mathcal{R}_k$ - Recurrent $k$-Temporal-Connectivity]
    \label{def:rktc} 
    The class $\mathcal{TC}^\mathcal{R}_k$ is the set of all evolving graphs $\mathcal{G}$ with infinite lifetime where $\forall j \in \mathcal{T}$, $\mathcal{G}_{[j,*]} \in \mathcal{TC}_k$ - the class where the temporal subgraph $\mathcal{G}_{[j,*]}$ is in $\mathcal{TC}_k$ for every time $j \in \mathcal{T}$.
\end{definition}

\begin{definition}[Class $\mathcal{C}^{*}_k$ - {1-interval $k$-connectivity}] 
	\label{def:ckinterval}
    The class $\mathcal{C}^{*}_k$ is the set of all evolving graphs $\mathcal{G}$ where $\forall G_j \in \mathcal{G}$, $G_j$ is a $k$-connected graph (the class where the node connectivity of every snapshot is at least $k$).
\end{definition}

\begin{definition}[Class $\mathcal{E}^{\mathcal{R}}_k$ (Recurrent Edges $k$-connected)]
	\label{def:reckedges}
    The class $\mathcal{E}^{\mathcal{R}}_k$ is the set of all evolving graphs $\mathcal{G}$ with infinite lifetime where, given $\mathbb{G} = (V,E)$ as underlying graph of $\mathcal{G}$, $\mathbb{G}$ is a k-connected graph and $\forall e \in  E, ~ \forall j \in \mathcal{T}, ~ \exists l > j , e \in E_l$ (the class where the underlying graph is $k$-connected and every edge is present infinitely often).
\end{definition}

\subsection{Authenticated Links}
\label{subsec:al}
Maurer et al.~\cite{DBLP:conf/srds/MaurerTD15} characterized the strict condition to solve a single one-to-one instance of the reliable communication problem in the \textit{perfect authenticated links} and \textit{synchronous computation} setting. {In other words, they identified the message propagation pattern that the processes and the communication network must support (i.e., the ``minimum'' set of journeys that a content must traverse) in order to solve a single instance of the reliable communication problem under the considered settings. We recall their contribution in Appendix \ref{sec:maurer}.


\begin{theorem}[{One-to-one RC at time $j$ in $\langle PL, SC, AL \rangle$}~\cite{DBLP:conf/srds/MaurerTD15}]
    \label{rm:maurer}
    The reliable communication problem can be solved,  starting at time $j$, from a defined source $p_s$ to a fixed target $p_t$, in the \textit{perfect authenticated links} and \textit{synchronous computation} setting, if and only if $\mathcal{G}_{[j,*]} \in \mathcal{J}_{(s,t,2f+1)}$. 
\end{theorem}

\noindent From this seminal result~\cite{DBLP:conf/srds/MaurerTD15}, we can derive the class $\mathcal{J}_{(s,t,k)}$ as the one that characterizes the communication networks where one-to-one reliable communication in $\langle PL, SC, AL \rangle$ can be solved at least once.

We generalize the result by Maurer et al.~\cite{DBLP:conf/srds/MaurerTD15} to any pair of processes and any time $j$. 

\begin{theorem}[RC at time $j$ in $\langle PL, SC, AL \rangle$]
    \label{th:anytoanysync}
    The reliable communication problem can be solved \underline{at time} $j$, in the \textit{perfect authenticated links} and  \textit{synchronous computation} setting, if and only if the dynamic subgraph $\mathcal{G}_{[j,*]} \in \mathcal{TC}_k$ and $k>2f$.
\end{theorem}

\begin{proof}
    The claim follows from an extension of Theorem \ref{rm:maurer}, by considering any pair of processes as the source and target of a reliable communication instance, and by the construction of class $\mathcal{TC}_k$ (Definition \ref{def:tempKconn}).
    Definition \ref{def:noname0} identifies the class of evolving graphs where the strict condition for one-to-one reliable communication at time $j$ holds (Theorem \ref{rm:maurer}).
    Given any four nodes $p_s,p_t,p_u,p_v \in V$, $\mathcal{G}_{[j,*]} \in \mathcal{J}_{(s,t,2f+1)}$ does not imply $\mathcal{G}_{[j,*]} \in \mathcal{J}_{(u,v,2f+1)}$; therefore, the condition must thus be verified for every pair of nodes.
    Definition \ref{def:tempKconn} extends Definition \ref{def:noname0} by considering any pair of processes; the claim thus follows given the strictness of Theorem \ref{rm:maurer}.
\end{proof}

\noindent The $\mathcal{TC}^\mathcal{R}_k$ class identifies the communication networks where the reliable communication problem is solvable in $\langle PL, SC, AL \rangle$ at any time $j\in\mathcal{T}$.

\begin{theorem}[RC in $\langle PL, SC, AL \rangle$]
    \label{rm:anytoanyAllSync}
    The reliable communication problem can be solved starting \underline{at any time} $j$, in the \textit{perfect authenticated links} and \textit{synchronous computation} setting, if and only if $\mathcal{G} \in \mathcal{TC}^\mathcal{R}_k$ and $k>2f$.  
\end{theorem}

\begin{proof}
    The claim follows from an extension of Theorem \ref{th:anytoanysync}, by considering any starting time for a reliable communication instance, and by the construction of class $\mathcal{TC}^\mathcal{R}_k$ (Definition \ref{def:rktc}).
    
    Definition \ref{def:tempKconn} identifies the class of evolving graphs where the strict condition for reliable communication at time $j$ holds (Theorem \ref{th:anytoanysync}). Given $x,y \in \mathcal{T}$, where $y > x$, $\mathcal{G}_{[y,*]} \in \mathcal{TC}_k$ implies $\mathcal{G}_{[x,*]} \in \mathcal{TC}_k$, but the opposite relation does not hold; thus, the condition needs to be extended to any time $j \in \mathcal{T}$.
    Definition \ref{def:rktc} extends Definition \ref{def:tempKconn} by considering any time $j \in \mathcal{T}$; the claim thus follows given the strictness of Theorem \ref{th:anytoanysync}.
\end{proof}

We identify further dynamic network sub-classes of $\mathcal{TC}^\mathcal{R}_k$, thus providing additional classes where the primitive is feasible.
We defined a sub-class of $\mathcal{E}^{\mathcal{R}}$ that relates to the $\mathcal{TC}^\mathcal{R}_k$ class and extends the result reported in the Theorem \ref{rm:anytoanyAllSync}.

\begin{theorem}
    \label{th:anytoanyrecurrent}
    Let $\mathcal{G}$ be an evolving graph with infinite lifetime. If there exists a spatial subgraph $\mathcal{G}' := \mathcal{G}[V, \bar{E}]$ of class $\mathcal{E}^{\mathcal{R}}_k$ then $\mathcal{G}$ is in $\mathcal{TC}^\mathcal{R}_k$.
\end{theorem}

\begin{proof}
    Consider a $k$-connected graph $\mathbb{G}' = (V, \bar{E})$ where $p_a$ and $p_b$ are two of its nodes. 
    It is known~\cite{Diestel_2017} that it is possible to identify a set of paths $p_a \rightarrow_k p_b$ between $p_a,p_b$ in $\mathbb{G}'$ such that its minimum cut is at least $k$ 
    (namely, there exists a set of paths $p_a \rightarrow_k p_b$ between $p_a,p_b$ in $\mathbb{G}'$ such that it is not possible to identify a subset $S \subset V \setminus \{p_a, p_b\}$ of size $k-1$ in which each path in $p_a \rightarrow_k p_b$ shares at least one node with $S$).
    If $\mathbb{G}'$ is the underlying graph of an evolving graph $\mathcal{G}'$ of class $\mathcal{E}^{\mathcal{R}}$ (thus $\mathcal{G}' \in \mathcal{E}^{\mathcal{R}}_k$) , then there always exists a set of journeys $p_a \rightsquigarrow_k p_b$ traversing the paths $p_a \rightarrow_k p_b$, because every edge re-appears infinitely often. 
    It follows that $\mathcal{G}' \in \mathcal{TC}^\mathcal{R}_k$ and the claim follows for any temporal graph $\mathcal{G}$ having as underlying graph $\mathbb{G} = (V, E \supseteq \bar{E})$. 
\end{proof}

\begin{corollary}[RC in $\langle PL, NC, AL \rangle$ - recurrent edges]
    \label{rm:anytoanyAllSync2}
    The reliable communication problem can be solved starting at any time $j$, in the \textit{perfect authenticated links} and \textit{synchronous computation} setting, if there exists a spatial subgraph $\mathcal{G}' := \mathcal{G}[{V}, \bar{E}]$ in $\mathcal{G}$ such that $\mathcal{G}' \in \mathcal{E}^{\mathcal{R}}_k$ and $k>2f$.   
\end{corollary}

\begin{proof}
    It follows from Theorems \ref{rm:anytoanyAllSync} and \ref{th:anytoanyrecurrent}. 
    Theorem \ref{rm:anytoanyAllSync} characterizes class $\mathcal{TC}^\mathcal{R}_k$ as the one where reliable communication is solvable at any time when $k>2f$; Theorem \ref{th:anytoanyrecurrent} identifies $\mathcal{E}^{\mathcal{R}}_k$ class as a subclass of $\mathcal{TC}^\mathcal{R}_k$. 
\end{proof}

Relations between recurrent reachability/connectivity and recurrent edges classes exist in both directions, as shown in the following.

\begin{theorem}
    \label{th:new1}
    If an evolving graph $\mathcal{G}$ is in $\mathcal{J}^\mathcal{R}_{(s,t)}$ then there exists a spatial subgraph $\mathcal{G}' := \mathcal{G}[\bar{V}, \bar{E}]$ of class $\mathcal{E}^{\mathcal{R}}$ having a path between $p_s$ and $p_t$ in its underlying graph $\bar{\mathbb{G}}=(\bar{V},\bar{E})$.
\end{theorem}

\begin{proof}
    The $\mathcal{J}^\mathcal{R}_{(s,t)}$ class guarantees the existence of a journey from node $p_s$ to node $p_t$ in all temporal subgraphs $\mathcal{G}_{[j,*]}$ of $\mathcal{G}$.

    Let $J_1 = (A_1,B_1)$ be a journey $(p_s \rightsquigarrow p_t) \in \mathcal{G}$, and let $x_{1}$ and $y_{1}$ be respectively the lowest and highest time instant associated to an edge in $J_1$, namely the first the and last elements of $B_1$.
    Let $J_2 = (A_2,B_2)$ be a journey $(p_s \rightsquigarrow p_t) \in \mathcal{G}_{[y_1+1,*]}$, and let $x_{2}$ and $y_{2}$ be respectively the lowest and highest time instant associated to an edge in $J_2$.
    The power set of a set $R$ is the set of all subsets of $R$. 
    The power set of a set that has a finite number of elements has a finite number of elements as well.
    It follows that, continuing the reasoning above, and thus identifying edge set $A_1, A_2, \dots$ there is at least one edge set $A_j$ that occurs infinitely often, and the claim follows.
\end{proof}

\begin{corollary}
    \label{th:new2}
    If an evolving graph $\mathcal{G}$ is in $\mathcal{J}_{(s,t,k)}^\mathcal{R}$ then there exists a spatial subgraph $\mathcal{G}' := \mathcal{G}[\bar{V}, \bar{E}]$ of class $\mathcal{E}^{\mathcal{R}}$ having a set of $k$ disjoint paths between $p_s$ and $p_t$ in its underlying graph $\bar{\mathbb{G}}=(\bar{V},\bar{E})$.
\end{corollary}

\begin{proof}
    The claim follows for the same argument provided in Theorem \ref{th:new1}.

    The $\mathcal{J}_{(s,t,k)}^\mathcal{R}$ class guarantees the existence of a set of journeys from $p_s$ to $p_t$, having a dynamic minimum cut of size at least $k$, for all temporal subgraphs $\mathcal{G}_{[j,*]}$ of $\mathcal{G}$, namely $\forall j \in \mathcal{T}, \exists (p_a \rightsquigarrow_k p_b) \in \mathcal{G}_{[j,*]}$.

    Let $P_1$ be a set of journeys $(p_a \rightsquigarrow_k p_b) \in \mathcal{G}$, let $x_{1}$ and $y_{1}$ be respectively the lowest and highest time instant associated to an edge of a journey in $P_1$, and let $E_1 \subseteq E$ be the set of edges in $P_1$. 
    Let $P_2$ be a set of journeys $(p_a \rightsquigarrow_k p_b) \in \mathcal{G}_{[y_1+1,*]}$, let $x_{2}$ and $y_{2}$ be respectively the lowest and highest time instant associated to an edge of a journey in $P_2$, and let $E_2\subseteq E$ be the set of edges in $P_2$.
    The power set of a set that has a finite number of elements has a finite number of elements as well.
    It follows that, continuing the reasoning above, and thus identifying edge set $E_1, E_2, \dots$, there is at least one edge set $E_j$ that occurs infinitely often, and the claim follows.
\end{proof}

\begin{corollary}
    \label{th:new3}
    Let $\mathcal{G}$ be an evolving graph with infinite lifetime. 
    If $\mathcal{G}$ is in $\mathcal{TC}^\mathcal{R}_k$ then there exists a spatial subgraph $\mathcal{G}' := \mathcal{G}[V, \bar{E}]$ of class $\mathcal{E}^{\mathcal{R}}_k$
\end{corollary}

\begin{proof}
    The claim follows for the same argument provided in Corollary \ref{th:new2}, considering any pair of nodes $p_s,p_t$.
\end{proof}

We prove that the classes $\mathcal{C}^{*}_k$ and $\mathcal{TC}^\mathcal{R}_k$ are also related. Specifically, $\mathcal{C}^{*}_k$ is a sub-class of $\mathcal{TC}^\mathcal{R}_k$ if $\mathcal{G}$ has infinite lifetime, and we accordingly extend previous solvability results on the reliable communication problem.

{\begin{theorem}
	\label{th:kinterval}
	Given an evolving graph $\mathcal{G}$ with infinite lifetime, if $\mathcal{G} \in \mathcal{C}^{*}_k$ then $\mathcal{G}_{[j,{j +n-k}]} \in \mathcal{TC}_k$ for any $j \in \mathcal{T}$, and thus $\mathcal{G} \in \mathcal{TC}^\mathcal{R}_k$.
\end{theorem}}

{
\begin{proof}    
    Consider an evolving graph $\mathcal{G} \in \mathcal{C}^{*}_k$ with $V$ as vertex set, $|V| = n$, and a pair of its nodes $p_s,p_t \in V$.

    Let $\Pi_{(s,t)}$ be the set of all the journeys from $p_s$ to $p_t$ in the temporal subgraph $\mathcal{G}_{[j,j +n-k]}$.
    Note that every temporal subgraph $\mathcal{G}_{\bar{T}}$ of an evolving graph $\mathcal{G}$ in $\mathcal{CK}^{*}_k$ is by definition also in  $\mathcal{CK}^{*}_k$.

    Let $S$ be a subset of $k-1$ nodes of $V$ not containing $p_s$ and $p_t$, namely $S \subset V \setminus \{p_s,p_t\}$, $|S| = k-1$, and 
    let $\mathcal{G}'$ be the spatial temporal subgraph of $\mathcal{G}$ such that $\mathcal{G}' := \mathcal{G}[V \setminus S]_{[j,{j+n-k}]}$.
    The evolving graph $\mathcal{G}'$ is 1-interval connected by construction because at least $k$ nodes must be removed from $\mathcal{G}$ to disconnect any of its snapshots. 

    Let $\Pi[V \setminus S]_{(s,t)}$ be the set of all the journeys from $p_s$ to $p_t$ in $\mathcal{G}'$.
    It has been proven in~\cite{DBLP:conf/stoc/KuhnLO10} that $n'-1$ instants, where $n'$ is the number of nodes in an evolving graph, are sufficient to traverse a journey between any two nodes in a 1-interval connected graph (class $\mathcal{C^*}$);
    the evolving graph $\mathcal{G}'$ is composed of $n-(k-1)$ nodes, thus a journey from $p_s$ to $p_t$ can always be traversed in at most $n-(k-1)-1 = n -k$ instants, and thus $\Pi[V \setminus S]_{(s,t)} \neq \emptyset$.
    
    It follows that there exists a set of journeys with a dynamic minimum cut size at least equal to $k$ in $\mathcal{G}_{[j,{j+n-k}]}$ from $p_s$ to $p_t$ 
    because there exists at least one journey in $\Pi_{(s,t)}$ from $p_s$ to $p_t$ when removing any subset $S \subset V \setminus \{p_s, p_t\}$ of $k-1$ nodes. 

    Given that, for any $j \in \mathcal{T}$,  $\mathcal{G}_{[j,{j +n-k}]} \in \mathcal{TC}_k$, it follows that $\mathcal{G} \in \mathcal{TC}^\mathcal{R}_k$ by Definition \ref{def:rktc}.
\end{proof}
}

\begin{corollary}[RC in $\langle PL, SC, AL \rangle$ - 1-interval]
    \label{cor:anytoany2}
    The reliable communication problem can be solved starting at any time $j$, in the \textit{perfect authenticated links} and \textit{negligible computation} setting, if $\mathcal{G} \in \mathcal{C}^{*}_k$ and $k>2f$. Furthermore, it can be solved in $n-k$ time.
\end{corollary}

\begin{proof}
    It follows from Theorems \ref{th:anytoanysync} and \ref{th:kinterval}.
    Theorem \ref{th:anytoanysync} characterizes class $\mathcal{TC}_k$ as the one where reliable communication is solvable when $k>2f$; 
    Theorem \ref{th:kinterval} identifies $\mathcal{C}^{*}_k$ class as a subclass of $\mathcal{TC}^\mathcal{R}_k$ where all of its temporal subgraphs with lifetime $n-k$ are in $\mathcal{TC}_k$. 
    The upper bound on the latency follows from the fact that $\forall j \in \mathcal{T}, \mathcal{G}_{[j,j+n-k]} \in \mathcal{TC}_k$, thus $p_s \rightsquigarrow_k p_t$ are traversable in $n-k$ times for whatever pair of $p_s$ and $p_t$.
\end{proof}

For the sake of completeness, we state a relation that exists between evolving graphs classes $\mathcal{C}^{*}_k$ and $\mathcal{E}^{\mathcal{R}}_k$.

\begin{theorem}
    \label{th:1intToRec}
    Let $\mathcal{G}$ be an evolving graph with infinite lifetime. If $\mathcal{G} \in \mathcal{C}^{*}_k$ (1-interval k-connectivity) then there exists a spatial subgraph $\mathcal{G}[V, E^\mathcal{R}]$ of class $\mathcal{E}^{\mathcal{R}}_k$ (recurrent edges k-connected).
\end{theorem}

\begin{proof}
    If an evolving graph $\mathcal{G}$ is 1-interval $k$-connected and has an infinite lifetime, then a subset of its edges  $E^\mathcal{R} \subseteq E$ must be present within an infinite number of snapshots. 
    Indeed, the number of nodes in $\mathcal{G}$ is finite and equals to $n$, and the number of possible edges is finite as well (at most $n^2$). It follows that some edges in $E$ must re-appear infinitely often. 
    We prove that if $\mathcal{G}$ is 1-interval $k$-connected then the set of edges $E^\mathcal{R}$ that re-appears infinitely often forms a $k$-connected graph $\mathbb{G}' = (V, E^\mathcal{R})$.
    
    Let us partition the edges of $E$ in $\mathcal{G}$ in two sets: $E^\mathcal{R}$ containing all the edges that re-appear infinitely often and $\tilde{E}$ that are present a finite number of times in $\mathcal{G}$.
    Let $t_z$ be the time when the last appearance of an edge in $\tilde{E}$ occurs in $\mathcal{G}$. It follows that starting at time $t_{z+1}$ all edges in $\mathcal{G}$ must appear infinitely often. The 1-interval $k$-connectivity property of $\mathcal{G}$ requires that all the edges $E^\mathcal{R}$ must form a $k$-connected graph, and the claim follows.
\end{proof}



Finally, we study the solvability conditions of the reliable communication problem while relaxing the assumptions of perfect links and synchronous computation.

\begin{theorem}[One-to-one RC in $\langle FLL / AC, AL \rangle$]
    \label{th:async10}
    Given a setting where either links are fair-loss, or local computation is asynchronous, or both, the one-to-one reliable communication problem can be solved starting at any time $j$ if and only if $\mathcal{G} \in \mathcal{J}_{(s,t,k)}^\mathcal{R}$ and $k>2f$.
\end{theorem}

\begin{proof}
    Provided in Appendix \ref{sec:results}. 
\end{proof}

\begin{theorem}[RC in $\langle FLL/AC, AL \rangle$]
    \label{th:async3}
    Given a setting where either links are fair-loss, or local computation is asynchronous, or both, the any-to-any reliable communication problem can be solved starting at any time $j$ if and only if $\mathcal{G} \in \mathcal{TC}^\mathcal{R}_k$ (Recurrent $k$-Temporal-Connectivity) and $k>2f$.
\end{theorem}

\begin{proof}
    The claim follow by extending arguments provided in Theorem \ref{th:async10} to any pair of nodes.
    Given any four nodes $p_s,p_t,p_u,p_v \in V$, $\mathcal{G} \in \mathcal{J}_{(s,t,k)}^\mathcal{R}$ does not imply $\mathcal{J}_{(u,v,k)}^\mathcal{R}$; therefore, the condition must thus be verified for every pair of nodes.
    Definition \ref{def:tempKconn} extends Definition \ref{def:noname0} by considering any pair of processes; the claim thus follows given the strictness of Theorem \ref{th:async10}.
\end{proof}

\noindent Different from the $\langle PL, SC, AL \rangle$ setting, the start time of the RC instance in the $\langle FLL/AC, AL \rangle$ case is irrelevant for the solvability conditions of the examined problems. Further details are reported in Corollaries 
8 and 
9 in Appendix \ref{sec:results}.

All the additional results identified for the defined sub-classes of $\mathcal{TC}^\mathcal{R}_k$ extend in the setting where either links are fair-loss, or local computation is asynchronous, or both.

\begin{corollary}[RC in $\langle FLL/AC, AL \rangle$]
    \label{th:async1}
    Given a setting where either links are fair-loss, or local computation is asynchronous, or both, the reliable communication problem can be solved starting at any time $j$ if $\mathcal{G} \in \mathcal{E}^{\mathcal{R}}_k$ (recurrent edges k-connected) and $k>2f$.
\end{corollary}

\begin{proof}
    It follows combining the results from Theorems \ref{th:anytoanyrecurrent} and \ref{th:async3}.
\end{proof}

\begin{corollary}[RC in $\langle FLL/AC, AL \rangle$]
    \label{anytoany3}
    Given a setting where either links are fair-loss, or local computation is asynchronous, or both, the reliable communication problem can be solved starting at any time $j$, if $\mathcal{G} \in \mathcal{C}^{*}_k$ (1-interval k-connectivity) and $k>2f$.
\end{corollary}

\begin{proof}
    It follows combining results from Theorems \ref{th:kinterval} and \ref{th:async3}.
\end{proof}

\subsection{Authenticated messages}
Theorem \ref{rm:maurer} by Maurer et al.~\cite{DBLP:conf/srds/MaurerTD15} identifies the base condition enabling one-to-one reliable communication from a defined source $p_s$ and a specific target $p_t$ starting at time $j$, that is, the existence of a set of journeys $p_s \rightsquigarrow_k p_t$ in $\mathcal{G}_{[j,*]}$ with $k>2f$. 
More specifically, there must exist a set of journeys $p_s \rightsquigarrow_{2f+1} p_t$ in $\mathcal{G}_{[j,*]}$ that cannot be cut by any set of $2f$ nodes~\cite{DBLP:conf/srds/MaurerTD15}. 
Maurer et al.~\cite{DBLP:conf/srds/MaurerTD15} additionally studied the one-to-one at time $j$ specification in the perfect links, synchronous computation and \textit{authenticated message} setting. 
We recall its solution in Appendix \ref{sec:results2} and the identified solvability condition in the following.

\begin{theorem}[One-to-one RC at time $t_j$ in $\langle PL, SC, AM \rangle$~\cite{DBLP:conf/srds/MaurerTD15}]
\label{th:maurer2}
    The one-to-one reliable communication problem can be solved starting at time $j$ from a process $p_s$ to a process $p_t$, in the \textit{perfect link}, \textit{synchronous computation}, and \textit{authenticated messages} setting, if and only if it exists a set of journeys $p_s \rightsquigarrow_k p_t$ in $\mathcal{G}_{[t_i,*]}$ and $k>f$.
\end{theorem}


\noindent All extension results are collected in the following Corollary.

\medskip
\begin{corollary}
[RC in $\langle *, *, AM \rangle$] 
    All the results available on the reliable communication problem for the authenticated link setting (AL), specifically Theorems \ref{th:anytoanysync}, \ref{rm:anytoanyAllSync}, \ref{th:async10}, \ref{th:async3}, and 11, 
    and Corollaries \ref{rm:anytoanyAllSync2}, \ref{cor:anytoany2},   \ref{th:async1}, \ref{anytoany3}, 
    8 and 
    9, extend to the authenticated message setting (AM) while requiring $k>f$ (instead of $2f$). 
More in detail, the solvability conditions to the reliable communication problem in the authenticated message settings are the ones reported in Table \ref{table:results}.
\end{corollary}


\begin{proof}
    The Corollary follows from the same argument given for the results presented in Subsection \ref{subsec:al}.
    Theorem \ref{th:maurer2} provides the network conditions for solving the one-to-one problem at a given time.
    The extension follows by identifying classes of evolving graphs for which such conditions are verified, considering every pair of processes and infinitely often occurrences of the conditions.
    The results differ on the parameter $k$, which must be greater than $f$.
    Note that the Maurer et al. solution reported in Algorithm 
    1 (Appendix \ref{sec:maurer}) needs to be modified as reported in Algorithm 2 (Appendix \ref{sec:results})
    to extend to fair-loss links and/or asynchronous computation settings.
\end{proof}
}

\subsection{On the complexity of verifying class membership}
The evolution of the communication network of a distributed system can be \textit{assumed}, in the sense that in certain real deployments it is reasonable to consider a particular model for the dynamic communication network. 
For example, it could be reasonable to assume that a swarm of mobile robots, moving inside a limited area, are repeatedly able to establish temporary communication links, and the resulting dynamic communication network provides recurrent temporal connectivity (class $\mathcal{TC}^\mathcal{R}$).
Alternatively, the \textit{complete characterization} of the evolution of a communication network, such as an evolving graph detailing the sets of available links over time, can be \textit{known in advance}. For example, considering the same example of a swarm of mobile robots, if the schedule of the exact movements of all the robots are known in advance, it is possible to deduce exactly when every pair of robot is able to establish a communication link.
In the latter setting, it is worth to notice that if a characterization of the communication network is provided as an evolving graph, the verification of class membership can be impractical to perform. More in detail, it has been proven~~\cite{DBLP:journals/tcs/FluschnikMNRZ20,DBLP:journals/jcss/KempeKK02,DBLP:journals/jcss/ZschocheFMN20} that it is NP-complete to decide whether the dynamic minimum cut size from a node $p_s$ to $p_t$ in an evolving graph is equal to a certain value $k$.
It follows that the solvability conditions presented in Theorems \ref{rm:maurer}, \ref{th:anytoanysync}, \ref{rm:anytoanyAllSync}, \ref{th:async3}, and 
11, and Corollaries 
8, and 
9, are NP-complete to verify on any temporal subgraph of the evolving graph.
On the other hand, conditions defined on the recurrent edges $k$-connected ($\mathcal{E}^{\mathcal{R}}_k$) and 1-interval $k$-connectivity ($\mathcal{C}^{*}_k$) classes can be verified with a polynomial algorithm on any temporal subgraph of an evolving graph, motivating the analysis of such subclasses of $\mathcal{TC}^\mathcal{R}_k$, particularly in light of the result presented in Corollary \ref{th:new3}.

\begin{figure}
\centering
\resizebox{.25\textwidth}{!}{%
\begin{tikzpicture}
  \draw [-stealth]
  node (TCR){$\mathcal{TC}^{\mathcal{R}}$}
  node (ER) [left=15pt of TCR] { $\mathcal{E}^{\mathcal{R}}$}
  node (CS) [above left=15pt of ER] { $\mathcal{C}^*$}
  node (CSK) [below = 55pt of CS, draw=blue] {\color{blue} $\mathcal{C}^*_k$}
  node (ERK) [below = 8pt of ER, draw=blue] {\color{blue} {$\mathcal{E}^{\mathcal{R}}_k$}}
  node (TCRK) [draw=red, below=15pt of TCR, circle]{\color{red}$\mathcal{TC}^{\mathcal{R}}_k$}

  (CS) edge[blue, dash pattern=on 1pt off 1pt, very thick] (ER)
  (ER) edge (TCR) 
  (CS) edge (TCR)
  (CSK) edge[blue, very thick] (CS)
  (CSK) edge[blue, dash pattern=on 1pt off 1pt, very thick] (ERK)
  (ERK) edge[blue, very thick] (ER)
  (ERK) edge[blue, very thick, bend left=25] (TCRK)
  (TCRK) edge[blue, dash pattern=on 1pt off 1pt, very thick] (ERK)
  (TCRK) edge[blue, very thick] (TCR)
  (CSK) edge[blue, very thick] (TCRK)
  ;
\end{tikzpicture}
 }%
 \caption{Relations between classes of evolving graphs. The classes and relations presented in this work are depicted in square bold blue, dashed edges represent the inclusion relation of a spatial subgraph. In the red circle is the minimal class of evolving graphs where the any-to-any reliable communication problem is solvable at any time under all the settings considered.}
 \label{fig:classes}
\end{figure}
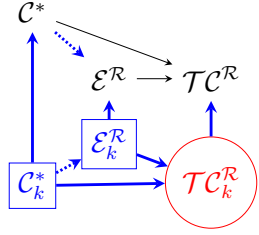

\section{Conclusion}
\label{sec:concl}



In this work, starting from the seminal contribution of Maurer et al.~\cite{DBLP:conf/srds/MaurerTD15}, we characterized the conditions that allow reliable communication between all processes at any time, and identified classes of dynamic networks that satisfy them. 
All the relations we identified between the classes of evolving graphs are summarized in Figure \ref{fig:classes}. 
In particular, class $\mathcal{TC}^\mathcal{R}_k$ is the smallest class of evolving graph for which the problem is solvable under all the settings we considered. Classes $\mathcal{E}^{\mathcal{R}}_k$ and $\mathcal{CK}^{*}_k$ are subclasses of $\mathcal{TC}^\mathcal{R}_k$ that can be verified in polynomial time on any temporal subgraph of an evolving graph.
All the identified results are outlined in Table \ref{table:results}.
\begin{table}
\centering
\caption{Strict Solvability Conditions.}
\begin{tabular}{l|c|}
\cline{2-2}
                                             & Solvability Conditions \\ \hline
\multicolumn{1}{|l|}{$\langle AL, PL, SC\rangle$}             &
\begin{tabular}{@{}c@{}}1-to-1 at $j$ : $\mathcal{J}_{(a,b,k)} \land k>2f$~\cite{DBLP:conf/srds/MaurerTD15}  \\ 1-to-1 at any $j$ : $\mathcal{J}_{(a,b,k)}^\mathcal{R} \land k>2f$ \\  *-to-* at $j$ : $\mathcal{TC}_k \land k>2f$ \\ *-to-* at any $j$ : $\mathcal{TC}_k^{\mathcal{R}} \land k>2f$\end{tabular}
\\ \hline
\multicolumn{1}{|l|}{
\begin{tabular}{@{}c@{}}$\langle AL, FLL, * \rangle$ \\ $\langle AL, *, AC \rangle$\end{tabular}} & 
\begin{tabular}{@{}c@{}}1-to-1 at (or at any) $j$ : $\mathcal{J}_{(a,b,k)}^\mathcal{R} \land k>2f$ \\  *-to-* at (or at any) $j$ : $\mathcal{TC}_k^{\mathcal{R}} \land k>2f$\end{tabular}
\\ \hline
\multicolumn{1}{|l|}{$\langle AM, PL, SC \rangle$}             & 
\begin{tabular}{@{}c@{}}1-to-1 at $j$ : $\mathcal{J}_{(a,b,k)} \land k>f$~\cite{DBLP:conf/srds/MaurerTD15}  \\ 1-to-1 at any $j$ : $\mathcal{J}_{(a,b,k)}^\mathcal{R} \land k>f$ \\  *-to-* at $j$ : $\mathcal{TC}_k \land k>f$ \\ *-to-* at any $j$ : $\mathcal{TC}_k^{\mathcal{R}} \land k>f$\end{tabular}
\\ \hline
\multicolumn{1}{|l|}{\begin{tabular}{@{}c@{}}$\langle AM, FLL, * \rangle$ \\ $\langle AM, *, AC \rangle $\end{tabular}} & 
\begin{tabular}{@{}c@{}}1-to-1 at (or at any) $j$ : $\mathcal{J}_{(a,b,k)}^\mathcal{R} \land k>f$ \\  *-to-* at (or at any) $j$ : $\mathcal{TC}_k^{\mathcal{R}} \land k>f$\end{tabular}
\\ \hline
\end{tabular}
\label{table:results}
\end{table}

Several interesting lines of future research include: compare our  deterministic classes of dynamic networks with those induced by probabilistic models, and identify (with high probability) equivalence conditions between the two models; analyze real datasets of dynamic networks (vehicles, drones, etc.), verifying whether identified conditions are satisfied; extend the study to open dynamic networks (where infinitely many processes may join and leave), or to the more demanding reliable broadcast problem~\cite{DBLP:conf/icdcs/BonomiDFRT21,DBLP:conf/opodis/BonomiFT23} (where the sender also can be Byzantine). 

\begin{credits}
\subsubsection{\ackname}
This work was partially supported by ANR project SAPPORO 2019-CE25-0005.

\subsubsection{\discintname} The authors have no competing interests to declare that are relevant to the content of this article. 
\end{credits}

\bibliographystyle{splncs04}
\bibliography{references}

\begin{thebibliography}{10}
\providecommand{\url}[1]{\texttt{#1}}
\providecommand{\urlprefix}{URL }
\providecommand{\doi}[1]{https://doi.org/#1}

\bibitem{DBLP:conf/icdcs/BonomiDFRT21}
Bonomi, S., Decouchant, J., Farina, G., Rahli, V., Tixeuil, S.: Practical byzantine reliable broadcast on partially connected networks. In: 41st {IEEE} International Conference on Distributed Computing Systems, {ICDCS} 2021, Washington DC, USA, July 7-10, 2021. pp. 506--516. {IEEE} (2021). \doi{10.1109/ICDCS51616.2021.00055}, \url{https://doi.org/10.1109/ICDCS51616.2021.00055}

\bibitem{DBLP:conf/sss/BonomiFT18}
Bonomi, S., Farina, G., Tixeuil, S.: Reliable broadcast in dynamic networks with locally bounded byzantine failures. In: Izumi, T., Kuznetsov, P. (eds.) Stabilization, Safety, and Security of Distributed Systems - 20th International Symposium, {SSS} 2018, Tokyo, Japan, November 4-7, 2018, Proceedings. Lecture Notes in Computer Science, vol. 11201, pp. 170--185. Springer (2018). \doi{10.1007/978-3-030-03232-6\_12}, \url{https://doi.org/10.1007/978-3-030-03232-6\_12}

\bibitem{DBLP:journals/jbcs/BonomiFT19}
Bonomi, S., Farina, G., Tixeuil, S.: Multi-hop byzantine reliable broadcast with honest dealer made practical. J. Braz. Comput. Soc.  \textbf{25}(1),  9:1--9:23 (2019). \doi{10.1186/s13173-019-0090-x}, \url{https://doi.org/10.1186/s13173-019-0090-x}

\bibitem{DBLP:conf/opodis/BonomiFT23}
Bonomi, S., Farina, G., Tixeuil, S.: Reliable broadcast despite mobile byzantine faults. In: Bessani, A., D{\'{e}}fago, X., Nakamura, J., Wada, K., Yamauchi, Y. (eds.) 27th International Conference on Principles of Distributed Systems, {OPODIS} 2023, December 6-8, 2023, Tokyo, Japan. LIPIcs, vol.~286, pp. 18:1--18:23. Schloss Dagstuhl - Leibniz-Zentrum f{\"{u}}r Informatik (2023). \doi{10.4230/LIPICS.OPODIS.2023.18}, \url{https://doi.org/10.4230/LIPIcs.OPODIS.2023.18}

\bibitem{10.1007/978-3-032-26465-7_7}
Bonomi, S., Farina, G., Tixeuil, S.: On the solvability of byzantine-tolerant reliable communication in dynamic networks. In: Georgiou, C. (ed.) Structural Information and Communication Complexity. pp. 112--130. Springer Nature Switzerland, Cham (2026). \doi{10.1007/978-3-032-26465-7_7}

\bibitem{BONOMI2024104952}
Bonomi, S., Farina, G., Tixeuil, S.: Reliable communication in dynamic networks with locally bounded byzantine faults. Journal of Parallel and Distributed Computing  \textbf{193},  104952 (2024). \doi{https://doi.org/10.1016/j.jpdc.2024.104952}, \url{https://www.sciencedirect.com/science/article/pii/S0743731524001163}

\bibitem{DBLP:books/daglib/0025983}
Cachin, C., Guerraoui, R., Rodrigues, L.E.T.: Introduction to Reliable and Secure Distributed Programming {(2.} ed.). Springer (2011). \doi{10.1007/978-3-642-15260-3}, \url{https://doi.org/10.1007/978-3-642-15260-3}

\bibitem{DBLP:books/hal/Casteigts18}
Casteigts, A.: A Journey through Dynamic Networks (with Excursions) (2018), \url{https://tel.archives-ouvertes.fr/tel-01883384}

\bibitem{DBLP:journals/paapp/CasteigtsFQS12}
Casteigts, A., Flocchini, P., Quattrociocchi, W., Santoro, N.: Time-varying graphs and dynamic networks. Int. J. Parallel Emergent Distributed Syst.  \textbf{27}(5),  387--408 (2012). \doi{10.1080/17445760.2012.668546}, \url{https://doi.org/10.1080/17445760.2012.668546}

\bibitem{castro1999practical}
Castro, M., Liskov, B.: Practical byzantine fault tolerance. In: Seltzer, M.I., Leach, P.J. (eds.) Proceedings of the Third {USENIX} Symposium on Operating Systems Design and Implementation (OSDI), New Orleans, Louisiana, USA, February 22-25, 1999. pp. 173--186. {USENIX} Association (1999), \url{https://dl.acm.org/citation.cfm?id=296824}

\bibitem{Diestel_2017}
Diestel, R.: Graph Theory. Springer Berlin Heidelberg (2017). \doi{10.1007/978-3-662-53622-3}

\bibitem{DBLP:conf/focs/Dolev81}
Dolev, D.: Unanimity in an unknown and unreliable environment. In: 22nd Annual Symposium on Foundations of Computer Science, Nashville, Tennessee, USA, 28-30 October 1981. pp. 159--168. {IEEE} Computer Society (1981). \doi{10.1109/SFCS.1981.53}, \url{https://doi.org/10.1109/SFCS.1981.53}

\bibitem{DBLP:conf/iptps/Douceur02}
Douceur, J.R.: The sybil attack. In: Peer-to-Peer Systems, First International Workshop, {IPTPS} 2002, Cambridge, MA, USA, March 7-8, 2002, Revised Papers. pp. 251--260 (2002). \doi{10.1007/3-540-45748-8\_24}

\bibitem{drabkin2005efficient}
Drabkin, V., Friedman, R., Segal, M.: Efficient byzantine broadcast in wireless ad-hoc networks. In: 2005 International Conference on Dependable Systems and Networks {(DSN} 2005), 28 June - 1 July 2005, Yokohama, Japan, Proceedings. pp. 160--169. {IEEE} Computer Society (2005). \doi{10.1109/DSN.2005.42}, \url{https://doi.org/10.1109/DSN.2005.42}

\bibitem{DBLP:journals/network/Ferreira04}
Ferreira, A.: Building a reference combinatorial model for manets. {IEEE} Netw.  \textbf{18}(5),  24--29 (2004). \doi{10.1109/MNET.2004.1337732}, \url{https://doi.org/10.1109/MNET.2004.1337732}

\bibitem{DBLP:journals/tcs/FluschnikMNRZ20}
Fluschnik, T., Molter, H., Niedermeier, R., Renken, M., Zschoche, P.: Temporal graph classes: {A} view through temporal separators. Theor. Comput. Sci.  \textbf{806},  197--218 (2020). \doi{10.1016/j.tcs.2019.03.031}, \url{https://doi.org/10.1016/j.tcs.2019.03.031}

\bibitem{DBLP:books/fm/GareyJ79}
Garey, M.R., Johnson, D.S.: Computers and Intractability: {A} Guide to the Theory of NP-Completeness. W. H. Freeman (1979)

\bibitem{DBLP:journals/jcss/KempeKK02}
Kempe, D., Kleinberg, J.M., Kumar, A.: Connectivity and inference problems for temporal networks. J. Comput. Syst. Sci.  \textbf{64}(4),  820--842 (2002). \doi{10.1006/jcss.2002.1829}, \url{https://doi.org/10.1006/jcss.2002.1829}

\bibitem{koo2004broadcast}
Koo, C.: Broadcast in radio networks tolerating byzantine adversarial behavior. In: Chaudhuri, S., Kutten, S. (eds.) Proceedings of the Twenty-Third Annual {ACM} Symposium on Principles of Distributed Computing, {PODC} 2004, St. John's, Newfoundland, Canada, July 25-28, 2004. pp. 275--282. {ACM} (2004). \doi{10.1145/1011767.1011807}, \url{https://doi.org/10.1145/1011767.1011807}

\bibitem{DBLP:conf/stoc/KuhnLO10}
Kuhn, F., Lynch, N.A., Oshman, R.: Distributed computation in dynamic networks. In: Schulman, L.J. (ed.) Proceedings of the 42nd {ACM} Symposium on Theory of Computing, {STOC} 2010, Cambridge, Massachusetts, USA, 5-8 June 2010. pp. 513--522. {ACM} (2010). \doi{10.1145/1806689.1806760}, \url{https://doi.org/10.1145/1806689.1806760}

\bibitem{DBLP:conf/opodis/Maurer20}
Maurer, A.: Self-stabilizing byzantine-resilient communication in dynamic networks. In: Bramas, Q., Oshman, R., Romano, P. (eds.) 24th International Conference on Principles of Distributed Systems, {OPODIS} 2020, December 14-16, 2020, Strasbourg, France (Virtual Conference). LIPIcs, vol.~184, pp. 27:1--27:11. Schloss Dagstuhl - Leibniz-Zentrum f{\"{u}}r Informatik (2020). \doi{10.4230/LIPIcs.OPODIS.2020.27}, \url{https://doi.org/10.4230/LIPIcs.OPODIS.2020.27}

\bibitem{DBLP:conf/wdag/MaurerT12}
Maurer, A., Tixeuil, S.: On byzantine broadcast in loosely connected networks. In: Aguilera, M.K. (ed.) Distributed Computing - 26th International Symposium, {DISC} 2012, Salvador, Brazil, October 16-18, 2012. Proceedings. Lecture Notes in Computer Science, vol.~7611, pp. 253--266. Springer (2012). \doi{10.1007/978-3-642-33651-5\_18}, \url{https://doi.org/10.1007/978-3-642-33651-5\_18}

\bibitem{DBLP:journals/jpdc/MaurerT14}
Maurer, A., Tixeuil, S.: Byzantine broadcast with fixed disjoint paths. J. Parallel Distributed Comput.  \textbf{74}(11),  3153--3160 (2014). \doi{10.1016/J.JPDC.2014.07.010}, \url{https://doi.org/10.1016/j.jpdc.2014.07.010}

\bibitem{DBLP:journals/tpds/MaurerT15}
Maurer, A., Tixeuil, S.: Containing byzantine failures with control zones. {IEEE} Trans. Parallel Distributed Syst.  \textbf{26}(2),  362--370 (2015). \doi{10.1109/TPDS.2014.2308190}, \url{https://doi.org/10.1109/TPDS.2014.2308190}

\bibitem{DBLP:journals/ppl/MaurerT16}
Maurer, A., Tixeuil, S.: Tolerating random byzantine failures in an unbounded network. Parallel Process. Lett.  \textbf{26}(1),  1650003:1--1650003:12 (2016). \doi{10.1142/S0129626416500031}, \url{https://doi.org/10.1142/S0129626416500031}

\bibitem{DBLP:conf/srds/MaurerTD15}
Maurer, A., Tixeuil, S., D{\'{e}}fago, X.: Communicating reliably in multihop dynamic networks despite byzantine failures. In: 34th {IEEE} Symposium on Reliable Distributed Systems, {SRDS} 2015, Montreal, QC, Canada, September 28 - October 1, 2015. pp. 238--245. {IEEE} Computer Society (2015). \doi{10.1109/SRDS.2015.10}, \url{https://doi.org/10.1109/SRDS.2015.10}

\bibitem{Menger1927}
Menger, K.: Zur allgemeinen kurventheorie. Fundamenta Mathematicae  \textbf{10}(1),  96--115 (1927), \url{http://eudml.org/doc/211191}

\bibitem{DBLP:journals/dc/PagourtzisPS17}
Pagourtzis, A., Panagiotakos, G., Sakavalas, D.: Reliable broadcast with respect to topology knowledge. Distributed Comput.  \textbf{30}(2),  87--102 (2017). \doi{10.1007/s00446-016-0279-6}, \url{https://doi.org/10.1007/s00446-016-0279-6}

\bibitem{DBLP:journals/networks/Pelc92}
Pelc, A.: Reliable communication in networks with byzantine link failures. Networks  \textbf{22}(5),  441--459 (1992). \doi{10.1002/net.3230220503}, \url{https://doi.org/10.1002/net.3230220503}

\bibitem{DBLP:journals/networks/Pelc96}
Pelc, A.: Fault-tolerant broadcasting and gossiping in communication networks. Networks  \textbf{28}(3),  143--156 (1996). \doi{10.1002/(SICI)1097-0037(199610)28:3<143::AID-NET3>3.0.CO;2-N}, \url{https://doi.org/10.1002/(SICI)1097-0037(199610)28:3<143::AID-NET3>3.0.CO;2-N}

\bibitem{DBLP:journals/ipl/PelcP05}
Pelc, A., Peleg, D.: Broadcasting with locally bounded byzantine faults. Inf. Process. Lett.  \textbf{93}(3),  109--115 (2005). \doi{10.1016/j.ipl.2004.10.007}, \url{https://doi.org/10.1016/j.ipl.2004.10.007}

\bibitem{DBLP:journals/wc/ZengGM10}
Zeng, K., Govindan, K., Mohapatra, P.: Non-cryptographic authentication and identification in wireless networks. {IEEE} Wireless Commun.  \textbf{17}(5),  56--62 (2010). \doi{10.1109/MWC.2010.5601959}

\bibitem{DBLP:journals/jcss/ZschocheFMN20}
Zschoche, P., Fluschnik, T., Molter, H., Niedermeier, R.: The complexity of finding small separators in temporal graphs. J. Comput. Syst. Sci.  \textbf{107},  72--92 (2020). \doi{10.1016/J.JCSS.2019.07.006}, \url{https://doi.org/10.1016/j.jcss.2019.07.006}

\end{thebibliography}

\appendix

\section*{Appendix}

\section{Maurer et al.~\cite{DBLP:conf/srds/MaurerTD15} Reliable Communication Results}
\label{sec:maurer}

Maurer et al.~\cite{DBLP:conf/srds/MaurerTD15} defined a protocol solving the reliable communication problem in the perfect authenticated links and synchronous computation setting.
The solution is provided in Algorithm \ref{alg:maurer}.
Every message exchanged by the protocol has the format $(s,c,S)$, where $s$ is the identifier of the source of the relayed content $c$, whereas $S$ is the set of nodes traversed by the content. 
Byzantine processes are able to generate messages with arbitrary fields.
Correct processes only deliver the contents that have been received through a collection of journeys that cannot be cut by $f$ nodes (where $f$ is a parameter of the protocol and supposed to be the upper bound of the total number of faulty nodes).

Initially, $\Omega_i$ in Algorithm \ref{alg:maurer} is initialized to the empty set.
When a process $p_s$ aims to reliably communicate a content $c$ at time $j$, it triggers the \mbox{\texttt{{RC.send($*$, $c$)}}} operation. 
This action adds $(s,c,\emptyset)$ to $\Omega_s$ and triggers \texttt{RC.deliver($s$, $c$)}\footnote{Notice that the reliable communication specification allows any correct process to reliable deliver a content as long as it has been sent by its source.}.
The solution is then defined by three rules: \textit{(i)} whenever the $\Omega_i$ or the set of neighbors changes, the process multicast $\Omega_i$; \textit{(ii)} upon receiving $\Omega_r$ from process $p_r$, the sender identifier $r$ is added to the set $S$ contained in each tuple $(s,c,S)$ in $\Omega_r$, and the resulting tuple is inserted in $\Omega_i$\footnote{Note that $S$ and $\Omega_i$ are sets; thus, the insertion of an already-contained element does not alter them.}; \textit{(iii)} whenever, among all the sets $S_x$ associated to a source $s$ and content $c$ in $\Omega_i$, it is possible to identify a minimum hitting set of size greater than $f$, then the content $c$ from $p_s$ is delivered.

\begin{algorithm}
	\caption{Maurer et al.~\cite{DBLP:conf/srds/MaurerTD15} Authenticated Link solution}
	\begin{algorithmic}[1] 

    \Statex Each correct node $p_i$ maintains the following variable:
\begin{itemize}
    \item $\Omega_i$, a dynamic set registering all the tuples $(s,c,S)$ received by the process;
\end{itemize}
\Statex Upon \texttt{{RC.send($*$, $c$)}} do
\begin{itemize}
    \item add $(s,c, \emptyset)$ to $\Omega_i$;
    \item \texttt{RC.deliver($s$, $c$)}.
\end{itemize}
\Statex 
Each correct process $p_i$ adheres the following rules (\emph{compute phase}):
\begin{enumerate}
    \item Whenever $\Omega_i$ or the set of neighbors changes, multicast $\Omega_i$;
    \item Upon reception of $\Omega_r$ through a link $\{p_r,p_i\}$, $\forall (s,c,S) \in \Omega_r$, add $(s,c, S \cup \{r\})$ to $\Omega_i$;
    \item Whenever there exists $s$, $c$ and a collection $\{S_1,\dots, S_l\}$ such that, $\forall x \in \{1,\dots,l\},(s,c,S_x \cup \{s\}) \in \Omega_i$ and $MinCut(\{S_1,\dots,S_l\}) > f$, trigger \texttt{RC.deliver($s$, $c$)}.
\end{enumerate}
		
	\end{algorithmic}
    \label{alg:maurer}
\end{algorithm}

\begin{lemma}[Necessary Condition to One-to-One RC at time $j$ in $\langle PL, SC, AL \rangle$~\cite{DBLP:conf/srds/MaurerTD15}]
    \label{lm:mausuff}
    For a given evolving graph $\mathcal{G}$, let us suppose that there exists an algorithm ensuring one-to-one reliable communication from $p_s$ to $p_t$ at time $j$. Then, we necessarily have $DynMinCut(\mathcal{G}_{[j,*]},p_s,p_t)>2f$.
\end{lemma}

\begin{lemma}[Sufficient Condition to One-to-one RC at time $j$ in $\langle PL, SC, AL \rangle$~\cite{DBLP:conf/srds/MaurerTD15}]
    \label{lm:maunec}
    For a given evolving graph $\mathcal{G}$, let $p_s$ and $p_t$ be two correct processes. If $DynMinCut(\mathcal{G}_{[j,*]},p_s,p_t)>2f$ then the Maurer et al. solution (Algorithm \ref{alg:maurer}) ensures one-to-one RC at time $j$ in $\langle PL, SC, AL \rangle$.
\end{lemma}

\begin{proof}
    The proofs for Lemmas \ref{lm:mausuff} and \ref{lm:maunec} are provided in~\cite{DBLP:conf/srds/MaurerTD15}.
    Given that several of the subsequent results follow from these seminal Lemmas, we provide an intuitive proof for the sufficient condition (Lemma \ref{lm:maunec}). 

    Let us refer with $\Pi_{s,t,j}(P)$ to the set of all the journeys in $\mathcal{G}_{[j,*]}$ passing through the path $P \in \mathbb{G}$ between $p_s$ and $p_t$.
    The assumption $DynMinCut(\mathcal{G}_{[j,*]},p_s,p_t)>2f$ implies that $\Pi_{s,t,j}(P)$ is not empty for at least one path $P \in \mathbb{G}$.
    We start by showing there is a correspondence between the journeys in $\mathcal{G}_{[j,*]}$ and the evolution of Algorithm \ref{alg:maurer}. Specifically, that if a process $p_s$ triggers \texttt{{RC.send($*$, $c$)}} at time $j$, then $\Omega_t$ will eventually contain $(s,c,P\setminus\{t\})$.
    For the sake of contradiction, let us assume that process $p_s$ triggers \texttt{{RC.send($*$, $c$)}} at time $j$, that the set $\Pi_{s,t,j}(P)$ is not empty, and that process $p_t$ will never store $(s,c,P \setminus \{t\})$ in $\Omega_t$.
    The \texttt{{RC.send($*$, $c$)}} trigger at time $j$ inserts $(s,c,\emptyset)$ to $\Omega_s$. 
    Such an insertion satisfies rule 1 of Algorithm \ref{alg:maurer}, thus if a link is present between a process $p_s$ and a process $p_u$ at time $j$, where $p_u$ is the second element of $P$, $(s,c,\emptyset)$ is sent to $p_u$, and $(s,c,\{s\})$ is then stored in $\Omega_u$.
    If the link $\{p_s,p_u\}$ is not available at time $j$ it will be at time $l>j$. Then, rule 1 of Algorithm \ref{alg:maurer} will be satisfied, and $(s,c,\{s\})$ will be stored in $\Omega_u$ at time $l$.
    If subsequently the insertion in $\Omega_u$ a link is present between $p_u$ and a process $p_v$, where $p_v$ is the third element of $P$, $(s,c,\{s\})$ is sent to $p_v$, and $(s,c,\{s,u\})$ is then stored in $\Omega_v$.
    The argument extends till the last element of $P$ leading to a contradiction.

    We proceed with proving the liveness property of reliable communication, thus showing that if $DynMinCut(\mathcal{G}_{[j,*]},p_s,p_t)>2f$ and $\texttt{{RC.send($*$, $c$)}}$ is triggered at time $j$, then \texttt{RC.deliver($s$, $c$)} will eventually be triggered by process $p_t$.
    The assumption $DynMinCut(\mathcal{G}_{[j,*]},p_s,p_t)>2f$ implies the existence of a set of journeys $Q = (p_s \rightsquigarrow_{>2f} p_t) \in \mathcal{G}_{[j,*]}$. 
    It follows from what has been proven above that $p_t$ will eventually store  $(s,c,P \setminus \{t\})$ in $\Omega_t$ for each path $P$ underlying a journey in $Q$. 
    The Byzantine processes can block only the propagation of tuples $(s,c,*)$ passing through them, thus at most $f$ journeys in $Q$ may not lead to the reception of the corresponding tuple by correct processes. 
    Therefore each correct process will eventually receive at least $f+1$ corresponding tuple among $Q$, ensuring eventual delivery of the content $c$.

    We conclude by proving the safety of reliable communication. 
    Let us assume that process $p_s$ does not trigger $\texttt{{RC.send($*$, $c$)}}$ and that Byzantine processes multicast tuples $(s,c,P)$, thus pretending content $c$ has been diffused by $p_s$.
    Each correct peer that receives $(s,c,P)$ from a Byzantine process $p_b$ adds the identifier $b$ of its sender to $P$ and then multicasts $(s,c,P \cup\{b\})$. 
    It follows that all the tuples $(s,c,P)$ stored by correct processes contain at least one identifier of a Byzantine process in $P$. 
    Given that Byzantine processes are at most $f$, the minimum hitting set size of all the received tuples associated to $(s,c)$ cannot exceed $f$, and the claim follows.
\end{proof}

\section{Intermediate Results and Missing Proofs}
\label{sec:results}
\subsection{Authenticated Links}
\subsubsection{One-to-one RC in $\langle PL, SC, AL \rangle$.}

The condition provided by Maurer et al.~\cite{DBLP:conf/srds/MaurerTD15} can easily be extended to consider any starting time $j$ in the class $\mathcal{J}_{(s,t,k)}^\mathcal{R}$.

\begin{theorem}[One-to-one RC in $\langle PL, SC, AL \rangle$]
    \label{rm:oneToOneAllSync}
    The one-to-one reliable communication problem can be solved starting at any time $j$, from a process $p_s$ to a process $p_t$, in the \textit{perfect authenticated links} and \textit{synchronous computation} setting, if and only if $\mathcal{G} \in \mathcal{J}_{(s,t,2f+1)}^\mathcal{R}$.  
\end{theorem}

\begin{proof}
    It follows from the construction of the $\mathcal{J}_{(s,t,k)}^\mathcal{R}$ class, extending Theorem \ref{rm:maurer}, that the condition identified by Maurer et al. is verified in $\mathcal{G}$ for any time $j$ according to the class specification.

    \textit{(Sufficient)} - Lemma \ref{lm:maunec} guarantees reliable communication from $p_s$ to $p_t$ at time $x$ if $\mathcal{G}_{[x,*]} \in \mathcal{J}_{(s,t,2f+1)}$; if $\forall j \in \mathcal{T}, \mathcal{G}_{[j,*]} \in \mathcal{J}_{(s,t,2f+1)}$ then it follows that the problem is solvable in the specific setting at any time $j$.

    \textit{(Necessary)} - Lemma \ref{lm:mausuff} demands $\mathcal{G}_{[x,*]} \in \mathcal{J}_{(s,t,2f+1)}$ to achieve reliable communication from $p_s$ to $p_t$ at time $x$. 
    A second instance between the same endpoints at time $y > x$ requires $\mathcal{G}_{[y,*]} \in \mathcal{J}_{(s,t,2f+1)}$ to be satisfied.
    $\mathcal{G}_{[y,*]} \in \mathcal{J}_{(s,t,2f+1)}$ implies $\mathcal{G}_{[x,*]} \in \mathcal{J}_{(s,t,2f+1)}$ when $y>x$, but the opposite relation does not hold. 
    The claim thus follows from the definition of class $\mathcal{J}_{(s,t,2f+1)}^\mathcal{R}$ by extending previous argument to any time $j \in \mathcal{T}$: $\mathcal{G}_{[j,*]} \in \mathcal{J}_{(s,t,2f+1)}$ will be required for any time $j$.
\end{proof}

\noindent The same conditions identified in Theorem \ref{rm:oneToOneAllSync} are necessary and sufficient to solve one-to-one reliable communication at time $j$, with the addition of transient failures~\cite{DBLP:conf/opodis/Maurer20}.

%
%

\subsubsection{Mauerer et al. protocol extension in FLL-AC.}

The Maurer et al. protocol has been proven correct to support reliable communication in the PL and SC setting. In order to extend its validity to weaker scenarios with FLL and AC, we modify the first of its rules as reported in Algorithm \ref{alg:edit}, while keeping the rest unchanged.
Fair-loss links do not guarantee single message deliveries but infinite deliveries in case of infinite transmissions.
It follows that, after a finite number of transmission attempts for a message $m$, $m$ is delivered for the first time. We thus edit Algorithm \ref{alg:maurer} accordingly in order to overcome such link behavior.

\begin{algorithm}
	\caption{Modification of the Maurer et al. solution.}
	\begin{algorithmic}[1] 

\Statex [...]

\begin{enumerate}
    \item Multicast $\Omega_i$ at every time $j$;
\end{enumerate}

\Statex [...]

	\end{algorithmic}
    \label{alg:edit}
\end{algorithm}

\noindent\textbf{Theorem \ref{th:async10} (One-to-one RC in $\langle FLL / AC, AL \rangle$).} 
    Given a setting where either links are fair-loss, or local computation is asynchronous, or both, the one-to-one reliable communication problem can be solved starting at any time $j$ if and only if $\mathcal{G} \in \mathcal{J}_{(s,t,k)}^\mathcal{R}$ and $k>2f$.

\begin{proof}
    \textit{(Sufficient)}: We show that Algorithm \ref{alg:edit} solves reliable communication in the considered setting.
    The AC assumption may delay the local compute phase on processes for an unknown but finite amount of time, namely, the execution of all three rules of Algorithm \ref{alg:edit} and the execution of the $\texttt{{RC.send($*$, $c$)}}$ operation may take a non-unitary amount of time. 
    We have from the FLL assumption that if a link is present infinitely often and a process sends a message through such a link repetitively, then such a message is received infinitely often by the link specification.
    It follows that if a process sends a message over an FLL link at time $j$, and at an infinite amount of subsequent times, then such a message will be received for the first time after a finite amount of time.

    Let us consider $P$ as a path between $p_s$ and $p_t$ in $\mathbb{G}$ whose edges re-appear infinitely often (its existence is guaranteed by Theorem \ref{th:new1}).
    We start by showing that if Algorithm \ref{alg:edit} is executed and process $p_s$ triggers $\texttt{{RC.send($*$, $c$)}}$ at time $j$, then process $p_t$ eventually stores $(s,c,P \setminus \{t\})$ in $\Omega_t$.
    The \texttt{{RC.send($*$, $c$)}} trigger at time $j$ inserts $(s,c,\emptyset)$ to $\Omega_s$. 
    The AC assumption may delay such an insertion for a finite amount of time $\delta$.
    Meanwhile, rule 1 is satisfied for the entire lifetime of the system. 
    It follows that, starting from time $j+\delta$, $(s,c,\emptyset)$ is enqueued to be multicast at every time instant, thus for an infinite amount of times.
    Again, the execution of such a rule may take an arbitrary but finite amount of time.
    The first edge $e$ of P, let's say $\{p_s, p_r\}$, will be present infinitely often. 
    It follows that: \textit{(i)} $(s,c,\emptyset)$ is enqueued to be sent infinitely often, and \textit{(ii)} the link $e$ will be present infinitely often. 
    Accordingly, after a finite amount of time, $(s,c,\emptyset)$ will be received by process $p_r$ for the first time.
    The argument extends till traversing all the processes in $P$, and generalizes to $k$ disjoint paths of recurrent edges available in $\mathcal{J}_{(s,t,k)}^\mathcal{R}$. For the same reasons provided in the proof of Lemma \ref{lm:maunec}, the claim follows.  

    \textit{(Necessary)}:
    Theorem \ref{rm:oneToOneAllSync} states that the same class $\mathcal{J}_{(s,t,2f)}^\mathcal{R}$ is necessary in a system model where PL and SC are assumed. The system model is weakened in the considered setting; indeed, links may lose messages, and computation may take arbitrary time. It follows that stricter necessary conditions cannot be found.
\end{proof}

\begin{corollary}[One-to-one RC at time $j$ in $\langle FLL/AC, AL \rangle$]
    \label{th:async20}
    Given a setting where either links are fair-loss, or local computation is asynchronous, or both, the one-to-one reliable communication problem can be solved starting at time $j$ if and only if $\mathcal{G} \in \mathcal{J}_{(a,b,k)}^\mathcal{R}$ (Recurrent k-journeys from $p_a$ to $p_b$) and $k>2f$.
\end{corollary}

\begin{corollary}[RC at time $j$ in $\langle FLL/AC, AL \rangle$]
    \label{th:async34}
    Given a setting where either links are fair-loss, or local computation is asynchronous, or both, the any-to-any reliable communication problem can be solved starting at time $j$ if and only if $\mathcal{G} \in \mathcal{TC}^\mathcal{R}_k$ (Recurrent $k$-Temporal-Connectivity) and $k>2f$.
\end{corollary}

\begin{proof}
    {\textit{(Sufficient)}: Theorems \ref{th:async10} and \ref{th:async3} define the enabling one-to-one and any-to-any conditions in the considered setting at any time $j\in\mathcal{T}$; thus, the conditions hold for a specific starting time.}

    {\textit{(Necessary)}}: For the same argument provided in Theorem \ref{th:async10}, the FLL and AC assumptions may not support message propagation for a limited but unknown amount of time, thus requiring the enabling condition to be verified indefinitely.
\end{proof}

\subsection{Authenticated Messages}
\label{sec:results2}

Maurer et al.~\cite{DBLP:conf/srds/MaurerTD15} provided another protocol solving the reliable communication problem in the authenticated message and synchronous computation setting.
The solution is provided in Algorithm \ref{alg:maurer2}.

The solution assumes an asymmetric encryption system is set up on the processes. 
Each process $p_i$ has a private key $priv_i$ (only known by $p_i$) and a public key $pub_i$ (known by all processes). 
The node $p_i$ can digitally sign a content $c$ with the function $sign(priv_i,c)$. 
Any node $p_t$ can verify content from $p_i$ with the function $verify(pub_i,c,\tilde{c})$. Note that $verify(pub_i,c,sign(priv_i,c)) = TRUE$.

\begin{algorithm}
	\caption{Maurer et al.~\cite{DBLP:conf/srds/MaurerTD15} Authenticated Message Solution}
	\begin{algorithmic}[1] 

    \Statex Each correct node $p_i$ maintains the following variable:
\begin{itemize}
    \item $\Omega_i$, a dynamic set registering all the tuples $(s,c,\tilde{c})$ received by the process;
\end{itemize}
\Statex Upon \texttt{{RC.send($*$, $c$)}} do
\begin{itemize}
    \item add $(s,c, sign(priv_i,c))$ to $\Omega_i$;
    \item \texttt{RC.deliver($s$, $c$)}.
\end{itemize}
\Statex 
Each correct process $p_i$ follows the following rules:
\begin{enumerate}
    \item Whenever $\Omega_i$ or the set of neighbors changes, multicast $\Omega_i$;
    \item Upon reception of $\Omega'$ through a link $\{p_x,p_i\}$, $\forall (s,c,\tilde{c}) \in \Omega'$, if $verify(pub_s,c,\tilde{c})$ then add $(s,c,\tilde{c})$ to $\Omega_i$;
    \item Whenever there exists $(s,c,\tilde{c}) \in \Omega_i$, trigger \texttt{RC.deliver($s$, $c$)}
\end{enumerate}
		
	\end{algorithmic}
    \label{alg:maurer2}
\end{algorithm}

Every message $m$ exchanged by the protocol has the format $(s,c,\tilde{c})$, where $s$ is the identifier of the source of the relayed content $c$, whereas $\tilde{c}$ is a digital signature. 
Byzantine processes are able to generate messages with arbitrary fields, but they are unable to generate a valid digital signature of content $c$ from a correct process.
A correct process only accepts contents that are relayed with a valid digital signature $\tilde{c}$, namely it accepts a content $c$ from $p_s$ only if it receives $\tilde{c}$ such that $verify(pub_s,c,\tilde{c}) = TRUE$.

\end{document}